\documentclass{article}
\pdfoutput=1

\usepackage{arxiv}
\usepackage{authblk} 

\usepackage[utf8]{inputenc} 
\usepackage[T1]{fontenc}    
\usepackage{hyperref}       
\usepackage{url}            
\usepackage{booktabs}       
\usepackage{amsfonts}       
\usepackage{nicefrac}       
\usepackage{microtype}      
\usepackage{lipsum}

\usepackage{verbatim}
\usepackage{mathrsfs}
\usepackage{amsthm}
\usepackage{blindtext}
\usepackage{tikz}
\usetikzlibrary{decorations.text,calc,arrows.meta}
\usetikzlibrary{backgrounds}
\usepackage{pifont}

\usepackage{amsfonts}
\usepackage{amsmath, amssymb}
\usepackage[]{algorithm2e}
\newtheorem{theorem}{Theorem}
\newtheorem{corollary}{Corollary}
\newtheorem{lemma}{Lemma}

\newtheorem{proposition}{Proposition}
\theoremstyle{definition}
\newtheorem{definition}{Definition}
\newtheorem{reduction}[]{Reduction}

\newcommand{\PW}{\mbox{PW}}
\newcommand{\NW}{\mbox{NW}}

\newcommand{\eat}[1]{}

\newcommand{\linearpartial}{\mbox{partial chain}}
\newcommand{\PSR}{\mbox{PSR}}

\newcommand{\PWPC}{\mbox{PW-PC}}
\newcommand{\PWPV}{\mbox{PW-PP}}
\newcommand{\PWDV}{\mbox{PW-DTB}}
\newcommand{\PWTV}{\mbox{PW-TTB}}
\newcommand{\PWBV}{\mbox{PW-BTB}}
\newcommand{\PTIME}{\mbox{\sc P}}
\newcommand{\NP}{\mbox{\sc NP}}
\newcommand{\maxpartial}[1]{s^{\text{max}}(#1, \boldsymbol{P}, \boldsymbol{Q})}
\newcommand{\Ptime}{\mbox{P}}

\title{The Complexity of Possible Winners on Partial Chains}

\author[1]{Vishal Chakraborty}
\author[1,2]{ Phokion G. Kolaitis}

\affil[1]{\footnotesize Department of Computer Science and Engineering, University of California, Santa Cruz}
\affil[2]{\footnotesize IBM Research - Almaden}

\date{February 26, 2020}
\begin{document}
\maketitle
\begin{abstract}
 The {\sc Possible Winner} ({\sc PW}) problem, a fundamental algorithmic problem in computational social choice, concerns elections where voters express only partial preferences between candidates. Via a sequence of investigations, a complete classification of the complexity of the PW problem was established for all pure positional scoring rules: the PW problem is in \PTIME~for the plurality and veto rules, and \NP -complete for all other such rules. More recently, the PW problem was studied on classes of restricted partial orders that arise in natural settings, such as partitioned partial orders  and truncated partial orders; in particular, it was shown that there are rules for which the PW problem drops from \NP -complete to \PTIME~on such restricted partial orders. Here, we investigate the PW problem on partial chains, i.e., partial orders that are a total order on a subset of their domains. Such orders arise naturally in a variety of settings, including rankings of movies or restaurants. We classify the complexity of the PW problem on partial chains by establishing that, perhaps surprisingly, this restriction does not change the complexity of the problem, namely, the PW problem is \NP -complete for all pure positional scoring rules other than the plurality and veto rules.  As a byproduct, we obtain a new and more principled proof of the complexity of the PW problem on arbitrary partial orders.
\end{abstract}


\section{Introduction}
Determining the winners in an election under various voting rules has been a mainstream topic of research in computational social choice. Ideally, each voter has a clear ranking among the candidates, from the most preferred one to the least preferred one. In reality, however, a voter may have only limited information about the candidates, which translates to the voter providing only a partial order among the candidates that reflects the voter's incomplete preferences (see the survey \cite{DBLP:reference/choice/BoutilierR16}). This state of affairs motivated \cite{konczak2005voting} to introduce the notion of \emph{possible winners} and \emph{necessary winners}, where a candidate is a possible (necessary) winner if the candidate is a winner in at least one (respectively, in all) sets of linear orders that extend the set of partial orders provided by the voters.

There has been an extensive study of the complexity of the associated decision problems {\sc Possible Winner} ({\sc PW}) and {\sc Necessary Winner} ({\sc NW}) with respect to a variety of voting rules.  Through a series of investigations  \cite{konczak2005voting,DBLP:journals/jair/XiaC11,DBLP:journals/jcss/BetzlerD10,DBLP:journals/ipl/BaumeisterR12}, the complexity of these problems has been  classified for all \emph{pure positional scoring rules} (see Section \ref{sec:prelim} for the precise definitions). Specifically,  {\sc NW} is in {\sc P} w.r.t.\ every pure positional scoring rule (where {\sc P} is the class of all decision problems solvable in polynomial time), while   {\sc PW} is in {\sc P} w.r.t.\ the plurality rule and the veto rule, but it is \NP-complete w.r.t.\ all other such rules. 

 More recently, the PW problem was studied on classes of restricted partial orders that arise in natural settings. For example, the preferences of the voters may be provided as \emph{top-truncated} partial orders, that is, partial orders in which each voter linearly orders  some top candidates, but  expresses no preference on  the rest. At the other end, we may have \emph{bottom-truncated} partial orders, where each voter linearly orders some bottom candidates (i.e., ``anybody but" candidates), but expresses no preference on the rest.
We may also have \emph{doubly-truncated} partial orders, where each voter linearly orders some top and bottom candidates, but expresses no preference for the ones  in the middle. In \cite{betzler2011unweighted,davies2011complexity,baumeister2012campaigns}, the complexity of the {\sc PW}~problem on such truncated partial orders was investigated. While no complete classification was obtained, it was shown that there are pure positional scoring rules, such as the $2$-approval rule, for which the complexity of {\sc PW} drops from \NP-complete to {\sc P}~on doubly-truncated partial orders. 

A partial order is \emph{partitioned} if its  elements can be partitioned into disjoint sets with a linear order between the disjoint sets, but no preference between elements in each set. 
In the
machine learning community, such partial orders were shown to be common in many real-life
datasets; furthermore, they have been used for learning statistical models on full and partial rankings \cite{lebanon2008non,lu2014effective,huang2012riffled}. Clearly, doubly-truncated partial orders are a special case of partitioned partial orders. In  \cite{DBLP:conf/atal/Kenig19},  the complexity of the {\sc PW} problem on partitioned partial orders was investigated and a nearly complete classification was obtained for positional scoring rules. In particular, it was shown that, for all $2$-valued rules (which contain $2$-approval as a special case) and also for the rule with scoring vectors of the form $(2,1,\ldots,1,0)$,  the  complexity of {\sc PW} on partitioned partial orders drops from \NP-complete to {\sc P}.

\paragraph{Summary of results}
In this paper, we  investigate  the  {\sc PW}  problem  on  \emph{partial chains}, i.e., partial orders that are a total order on a subset of their domains.  Such orders arise naturally in elections in which the number of candidates is large and, as a result, each voter can rank only a subset of the candidates. For example, consider the movies released in 2019.  Most viewers have seen only a subset of these movies and so they can only rank the movies they have seen. A similar state of affairs holds for songs, books, restaurants, and so on.
Partial chains are the most fitting model for this type of scenario. Indeed, it might be the case that a voter will like a movie they have not seen so far more than any of the the movies they have already seen (or less than any of the the movies they have already seen). This state of affairs can be modelled by partial chains, but not by partitioned, doubly-truncated, top-truncated, or bottom-truncated partial orders. 

We obtain a complete classification of the complexity of  the {\PW} problem on partial chains by establishing that this restriction does not change the complexity of the problem, namely, {\PW} is \NP-complete for all pure positional scoring rules other than the plurality rule and the veto rule.  This result should be contrasted with the aforementioned results about the drop in complexity of {\sc PW} on doubly-truncated partial orders and on partitioned partial orders. Our result also yields a new, self-contained proof of the classification of the complexity of {\sc PW} on arbitrary partial orders. Moreover, unlike the proof of the original classification theorem, our proof uses reductions from a single \NP-complete problem, namely, the {\sc 3-Dimensional Matching} problem.
  
Finally, we obtain new results about the complexity of the {\sc PW} problem on doubly-truncated partial orders by establishing that this problem is \NP-complete for a variety of pure positional scoring rules that were not covered by the earlier work on this problem. These rules include a broad group of both $p$-valued rules as well as unbounded rules.

\section{Preliminaries and Earlier Work} \label{sec:prelim}
\paragraph{Voting profiles}
A \emph{(strict) partial order} on a set $C$ is a binary relation $\succ$ on $C$ that is irreflexive (i.e., $a \not \succ a$, for every $a\in C$) and transitive (i.e., $a \succ b$ and $b \succ c$ imply $a \succ c$, for all $a,b,c \in C$. A \emph{total order} on $C$ is a partial order $\succ$ on $C$ such that for all $a,b\in C$, we have $a= b$ or $a \succ b$ or $ b \succ a$. 

Let $C = \{c_1, \ldots, c_m\}$ be a set of \emph{candidates} and let $V = \{v_1,\ldots,v_n\}$ be a set of voters.
A \emph{(complete) voting profile} is a tuple ${\bf T}=(T_1,\ldots,T_n)$ of total orders on elements of $C$, where each $T_l$ represents the ranking (preference) of voter $v_l$ on the candidates in $C$.
 Similarly, a \emph{partial voting profile} is a tuple
${\bf P}=(P_1,\ldots,P_n)$ of partial orders on $C$, where each $P_l$ represents the partial preferences of voter $v_l$ on the candidates in $C$.   
A \emph{completion} of a partial voting profile
${\bf P}= (P_1,\ldots,P_n)$ is a complete voting profile ${\bf T}=
(T_1,\ldots,T_n)$ such that each $T_l$ is a completion of the partial
order $P_l$, i.e., $T_l$ is a total order that extends $P_l$. Note that a  partial voting profile may have exponentially many completions.

\paragraph{Voting rules} We focus on \emph{positional scoring rules}, a widely studied class of voting rules.  A positional scoring rule $r$ on a set of $m$
candidates is specified by a scoring vector ${\bf s}=(s_1,\ldots,s_m)$
of non-negative integers, called the \emph{score values}, such that
$s_1\geq s_2\geq \ldots \geq s_m$ and $s_1 > s_m$.
Suppose that ${\bf
  T}=(T_1,\ldots,T_n)$ is a total voting profile. The score $s(T_l,c)$
of a candidate $c$ on $T_l$ is the score value $s_k$ where $k$ is the
position of candidate $c$ in $T_l$.  The \emph{score} of $c$ under the positional scoring rule $r$ on the total profile ${\bf T}$ is the  sum $\sum_{l=1}^ns(T_l,c)$. A candidate $c$ is a \emph{winner} if $c$'s score is greater than or equal to the scores of all other candidates; similarly, $c$ is a \emph{unique winner} if $c$'s score is greater than the scores of all other candidates. The set of all winners is denoted by
 $\mbox{W}(r, {\bf T})$.

We consider positional scoring rules that are defined for
every number $m$ of candidates. Thus, a \emph{positional scoring rule}
is an infinite sequence $\boldsymbol{s}_1, \boldsymbol{s}_2, \ldots, \boldsymbol{s}_m,
\ldots$ of scoring vectors such that each $\boldsymbol{s}_m$ is a scoring
vector of length $m$. Alternatively, a positional scoring rule is a
function $r$ that takes as argument a pair $(m,s)$ of positive
integers with $s\leq m$ and returns as value a non-negative integer
$r(m,s)$ such that $r(m,1) \geq r(m,2) \ldots \geq r(m,m)$. We assume that the function $r$ is computable in time polynomial in $m$, hence
 the winners can be computed in polynomial time.
Such a rule is \emph{pure} if the scoring vector ${\bf s}_{m+1}$ of
length $(m+1)$ is obtained from the scoring vector $\boldsymbol{s}_m$ of
length $m$ by inserting a score value in some position of 
$\boldsymbol{s}_ m$, provided that the non-increasing order of score values is
maintained. For every  scoring rule $\boldsymbol{s}_m$, multiplying all score  values by the same value, and adding the same constant to all  score values  does not change the winners; thus,   we assume that the $s_1, \hdots, s_m$ are co-prime and that there exists a $k$ such that $s_j = 0$ for all $j > k$. Such a $\boldsymbol{s}_m$ is called a \emph{normalised} scoring vector. Note that this is not a restriction. 
The  plurality rule $(1,0,\ldots, 0)$, the veto rule
$(1,\ldots,1,0)$, the $t$-approval rule $(\underbrace{1,\ldots, 1}_t, 0,\ldots, 0)$ with a fixed $t\geq 2$, for $m >2$, and the Borda count $(m-1,m-2, \ldots, 1, 0)$  are  prominent pure positional scoring rules.

\paragraph{Necessary and possible winners} 
Let $r$ be a voting rule and ${\bf P}$ a partial voting profile. 
The following notions 
were introduced by Konczak and
Lang~\cite{konczak2005voting}.

\begin{itemize}

\item 
The set $\PW(r,{\bf P})$ of the \emph{possible winners} w.r.t.\
 $r$ and $\bf P$ is the union of the sets $\text{W}(r,{\bf T})$, where
$\bf T$ varies over all completions of $\bf P$.  Thus, a
candidate $c$ is a \emph{possible winner} w.r.t.\ $r$ and $\bf
P$, if $c$ is  in the set $\text{W}(r,{\bf T})$ of winners, for at least one completion
$\bf T$ of $\bf P$.

The {\sc Possible Winner problem} ({\sc PW}) w.r.t.\  $r$ asks: given a set of candidates $C$, a partial profile $\bf P$, and a distinguished candidate $c \in C$, is $c \in  \PW(r,{\bf P})$?

\item 
The
set $\NW(r,{\bf P})$ of the \emph{necessary winners} w.r.t.\
$r$ and $\bf P$ is the intersection of the sets $\text{W}(r,{\bf T})$, where
$\bf T$ varies over all completions of $\bf P$. Thus, a
candidate $c$ is a \emph{necessary winner} w.r.t.\ $r$ and
$P$, if $c$ is  in the set  $\text{W}(r,{\bf T})$ of winners, for every completion $\bf
T$ of $\bf P$.

The {\sc Necessary Winner problem} ({\sc NW}) w.r.t.\  $r$ asks: given a set of candidates $C$, a partial profile $\bf P$, and a distinguished candidate $c \in C$, is $c \in  \NW(r,{\bf P})$?

\end{itemize}

The notions of \emph{necessary unique winners} and \emph{possible unique winners} are defined in an analogous manner.

Through the initial investigation by
\cite{konczak2005voting}~and subsequent investigations by~\cite{DBLP:journals/jair/XiaC11}, 
\cite{DBLP:journals/jcss/BetzlerD10}, and~\cite{DBLP:journals/ipl/BaumeisterR12}, the following
  classification of the complexity of the necessary and
the possible winners  for \emph{all} pure positional scoring rules was
established.

\begin{theorem} \label{class-thm} {\rm[Classification Theorem]}
The following hold.
\begin{itemize}
\item 
%
For every pure positional scoring rule $r$, the necessary winner problem \emph{\NW} w.r.t.\ $r$ is in \emph{\PTIME}.
\item The possible winner problem
\emph{\PW}  w.r.t.\ the plurality rule and the veto rule is in \emph{\PTIME}.   For
  all other pure positional scoring rules $r$, this problem  is \emph{\NP}-complete. 
\end{itemize}
Furthermore, the same classification holds for necessary unique winners and possible unique winners.
\end{theorem}

The proof of the above classification is rather involved; also, it is not self-contained as it spans several papers. The  proofs of NP-hardness for various positional scoring rules use reductions from 
several different known \NP -complete problems, including  {\sc $3$-Dimensional Matching}, {\sc Exact $3$-Cover}, {\sc Hitting Set},  {\sc $3$-SAT}, and {\sc Multicoloured Cliques}.
\section{Complexity of PW on Partial Chains} \label{sec:results}
This section contains the main result of the paper. We begin by defining the concept of a partial chain.

\begin{definition}
A partial order on a set $C$ is  a \emph{\linearpartial} if it is a linear order on a non-empty subset $C'$ of $C$.
\end{definition}

Let $C = \{a, b, c, d, e \}$ be a set of candidates. Clearly, every total order on $C$ is  a partial chain. Two other examples of partial chains on $C$ are  $a \succ d \succ c$  and $d \succ a \succ c \succ b$.

\begin{definition}
We write \PWPC~to denote the restriction of the \PW \ problem to partial chains. More precisely, the \PWPC~problem asks: given a set of candidates $C$, a partial profile $\bf P$ in which every partial order $P_l$, $1\leq l \leq n$, is a partial chain, and a distinguished candidate $c \in C$, is $c \in  \PW(r,{\bf P})$?

\end{definition}

Since \PWPC~is a special case of \PW, Theorem \ref{class-thm} implies that if $r$ is the plurality rule or the veto rule, then the \PWPC~problem with respect to $r$ is in \PTIME. 
The main result of this paper asserts that these are the only tractable cases.

\begin{theorem}\label{PWPC-hard-thm}
Let $r$ be a pure positional scoring rule other than the plurality and the veto rules. Then the \emph{\PWPC}~problem with respect to $r$ is \emph{\NP}-complete.

\end{theorem}

\begin{corollary} \label{PWPC-class-thm} {\rm[Classification Theorem for Partial Chains]}
\begin{enumerate}
\item If $r$ is plurality rule or the veto rule, then the \emph{\PWPC} problem with respect to $r$  is in \emph{\PTIME}.

\item For
  all other pure positional scoring rules $r$, the \emph{\PWPC} problem with respect to $r$ is \emph{\NP}-complete.
\end{enumerate}
\end{corollary}
\subsection{Proof outline of Theorem \ref{PWPC-hard-thm}}

\paragraph{\NP-complete problem used}
As mentioned earlier, the \NP-completeness of \PW~for rules other than plurality and veto in Theorem \ref{class-thm} was established via reductions from a variety of well known \NP-complete problems. Furthermore, none of these reductions used partial chains in the \PW-instances constructed. Here, we will establish the \NP-hardness of \PWPC~for rules other than plurality and veto via reductions from a \emph{single} well known \NP-complete problem, namely,  the 
{\sc 3-Dimensional Matching (3DM)} Problem
(Problem [SP1]  in \cite{garey1979guide}). This problem asks:
given three disjoint  sets $\mathcal{X} = \{x_1, \hdots, x_q\},~ \mathcal{Y} = \{y_1, \hdots, y_q\},~ \mathcal{Z}= \{z_1, \hdots, z_q\}$ of the same size, and a set   $\mathscr{S} \subseteq \mathcal{X} \times \mathcal{Y} \times \mathcal{Z}$, 
is there a subset $\mathscr{S}' \subseteq \mathscr{S}$ 
such that $| \mathscr{S}'| = q$ and $\mathscr{S}'$ does not contain two different triples that agree in at least one of their coordinates?

\paragraph{Grouping of Pure Positional Scoring Rules} The
 NP-hardness of \PW~with respect to rules other than plurality and veto in Theorem \ref{class-thm} was established by considering either groups of rules with similar characteristics \cite{DBLP:journals/jair/XiaC11} or individual rules, e.g., the rule
 with scoring vectors of the form $(2,1,\ldots,1,0)$
 \cite{DBLP:journals/ipl/BaumeisterR12}. Here, 
 we will establish the NP-hardness of \PWPC~with respect to rules other than plurality and veto  by grouping the pure positional scoring rules into two different groups, namely, \emph{bounded} rules
 and \emph{unbounded} rules.

\eat{
We show \NP-completeness using reductions from  The problem is defined as follows.

\begin{definition}
 Given disjoint  sets $\mathcal{X} = \{x_1, \hdots, x_q\}, \mathcal{Y} = \{y_1, \hdots, y_q\}, \text{ and } \mathcal{Z}= \{z_1, \hdots, z_q\}$, a family of triples, $\mathscr{S} \subseteq \mathcal{X} \times \mathcal{Y} \times \mathcal{Z}  = \{S_1, S_2, S_3, \hdots, S_t\}$, it asks whether there is a subset $\mathscr{S}' \subseteq \mathcal{S}$ 
such that $| \mathscr{S}'| = q$ and no co-ordinates of the triples in $\mathscr{S}'$ are identical.
\end{definition}

Define for $x \in \mathcal{X} \cup \mathcal{Y} \cup \mathcal{Z}$, $n_x \in \mathbb{N} = |\{S_i \mid x \in S_i \} |.$ This problem, even with the restriction that $\forall x \in \mathcal{X} \cup \mathcal{Y} \cup \mathcal{Z}, n_x \in \{2,3\}$ is know to be is \NP-complete \footnote{A more general version of the problem where the cardinality of $\mathcal{X}$, $\mathcal{Y}$, and $\mathcal{Z}$ are not equal was shown to be \NP-complete in \cite{karp1972reducibility}} [problem [SP1] page 221 in \cite{garey1979guide}].

In the reductions to prove \NP-completeness, we exploit specific properties of \PSR s. These properties give rise to a natural categorisation of \PSR s. First, we present the following definition.
}

\begin{definition} \label{bdd-value-defn}
Let $r$ be a pure positional scoring rule.
\begin{itemize}
    \item We say that $r$ is
    \emph{$p$-valued}, where $p$ is a positive integer greater than $1$,  if there exists a positive integer $n_0$ such that for all $m \geq n_0$, the scoring vector $\boldsymbol{s}_m$ of $r$ contains exactly $p$ distinct values. 
    \item We say that $r$ is \emph{bounded} if $r$ is $p$-valued, for some $p >1$; otherwise, $r$ is \emph{unbounded}.
\end{itemize}
\end{definition}

Clearly, the plurality rule, the veto rule, and the $t$-approval rule with fixed $t\geq 2$, are $2$-valued rules. For a different example of a $2$-valued rule, consider the rule with scoring vectors 
${\bf s}_{2m}=(\underbrace{1,\ldots,1}_m,\underbrace{0,\ldots,0}_m)$ and 
${\bf s}_{2m+1}=(\underbrace{1,\ldots,1}_{m+1},\underbrace{0,\ldots,0}_m)$, where $m\geq 1$. Furthermore,
the rule with scoring vectors of the form $(2,1,\dots,1,0)$ is $3$-valued, while  the Borda count $(m-1,m-2,\ldots,0)$ is an unbounded rule.
Note also that, unlike the Borda count, an unbounded scoring rule may have score values that are not decreasing at the same rate or may have arbitrarily long repeating score values. 

\paragraph{Main Steps}
The technical cornerstones of the proof of Theorem \ref{PWPC-hard-thm}  are three polynomial-time reductions, each of which reduces the {\sc 3DM} problem to the \PWPC~problem with respect to the following types of pure positional scoring rules: 
\begin{itemize}
    \item $2$-approval, which is then extended to all $2$-valued rules other than plurality and veto.
    \item $3$-valued rules, which is then extended to all $p$-valued rules with $p> 3$.
    \item unbounded scoring rules.
\end{itemize}

In each reduction, the partial profile we construct from an arbitrary {\sc 3DM} instance consists of two parts. The first part is a set of \linearpartial s  (which are not total orders). These encode the given instance of the {\sc 3DM} problem. It is worth pointing out that these partial chains have \emph{at most two}  candidates ``missing".
The high-level idea of the construction is as follows. In order for candidate $c$ to win in some  completion of the partial chains, some  other candidates have to lose points. Suppose $c'$ is one such candidate. To lose points, $c'$ has to be in a higher position. Whenever $c'$ is in a higher position, a few other candidates are ``pushed up" to lower positions, and they gain points. The score of these candidates are set in such a way that they can be ``pushed up" only once. We set the specific scores for every candidate using the second part of the partial profile, which consists of a total profile. These votes, which fulfil certain properties, can be constructed in time polynomial in the number of candidates due to a result similar to the one in \cite[Lemma 4.2]{baumeister2011computational}. A variant of this result has been used in the literature \cite{dey2016exact,dey2016complexity,dey2016kernelization,dey2018complexity,DBLP:conf/atal/Kenig19}. To make our work self-contained, we state and prove the following variant of the original result and use it in all our reductions from {\sc 3DM} to \PWPC.
\begin{lemma} \label{lemmaDM}
Given a set
$C = \{ c_1 , \hdots, c_{m} \}$ of candidates, a singleton
$D = \{d\}$, 
a normalised scoring vector $\mathbf{s}$ of length $m+1$,
and for every $c_i$, a list of integers $\eta_{i,1}, \hdots, \eta_{i,m}$ with $\sum\limits_{j=1}^{m} | \eta_{i,j} | \leq O(m^4),$ one can construct, in time polynomial in $m$, a total voting profile $\boldsymbol{Q}$ and a  $\lambda_{\boldsymbol{Q}} \in \mathbb{N}$ such that, for $1 \leq i \leq m,$ the score  $s(\boldsymbol{Q}, c_i) = \lambda_{\boldsymbol{Q}} + R_i$ where $R_i = \sum\limits_{j=1}^{m} \eta_{i,j} (s_j - s_{j+1})$ 
and $s(\boldsymbol{Q}, d) < \lambda_{\boldsymbol{Q}}$. 
In particular, the number of votes in the profile $\boldsymbol{Q}$ is polynomial in $m$.
\eat{ 
$\cdot \sum\limits_{i=1}^{m} | R_i |.$ 
}
\end{lemma}
\begin{proof}
Before proving the lemma, we introduce some notation which will be useful for the proof. Let, for $1 \leq j \leq m$, the value $\delta_{j} = s_j - s_{j+1}$.
 For $1 \leq j \leq m+1$, let $\beta(d, j)$ be a block of $|C| = m$ votes where
\begin{itemize}
    \item in each vote in the block, $d$ is in position $j$.
    \item in the $m$ votes, all the candidates besides $d$ are in each position exactly once.
\end{itemize}
 For a given $j$, note that there are many ways to construct such a block $\beta(d, j)$. Given $j$, we fix a block $\beta(d, j)$. The following is an example of the block $\beta(d, 1)$.
\begin{align*}
    d &>& c_1       &>& c_2 &>& \hdots &>& c_{m-1} &>& c_m \\
    d &>& c_2       &>& c_3 &>& \hdots &>& c_{m} &>& c_{1} \\
     d &>& c_3       &>& c_4 &>& \hdots &>& c_{1} &>& c_{2} \\
     \vdots & & \vdots       & & \vdots & & \vdots & & \vdots & & \vdots \\
    d &>& c_{m}     &>& c_1 &>& \hdots &>& c_{m-2} &>& c_{m-1} \\
\end{align*}
When we talk about the score of a candidate in a block, we refer to the score of the candidate in a profile containing only the votes of that block. Observe that no matter how a block is constructed, the score of all the candidates is always the same. The score of the candidates in the block $\beta(d, j)$ are as follows.
\begin{itemize}
    \item $s(\beta(d,j), d) = m s_j$
    
    \item For all $ x \in C \setminus \{d\}, \text{ we have }s(\beta(d,j), x) = - s_j +  \sum\limits_{k=1}^{m+1} s_k = \lambda_{\beta(d,j)}$
\end{itemize}
If candidate $d$ in position $j$ and candidate $c$ is in position $j+1$ are swapped, the score of $c$ \textit{increases} by $\delta_j,$ i.e., the score of $c$ is $ \lambda_{\beta(d,j)} + \delta_j.$ The score of $d$ decreases and the scores of all the candidates in $C \setminus \{d,c\}$ remain unchanged. We will use this idea construct the total profile $\boldsymbol{Q}.$
\\
Let $\boldsymbol{Q}$ be an empty profile and $\lambda_{\boldsymbol{Q}} = 0.$ We will construct the total profile $\boldsymbol{Q}$ incrementally. 
For each $c_i \in C$, for each $\eta_{i,j}$, where $1 \leq j \leq m$, and $\eta_{i,j} \neq 0$, we add votes to $\boldsymbol{Q}$ in the following two steps.
\begin{enumerate}
        \item This consists of two cases. 
        \begin{description}
            \item[Case I.] $\eta_{i,j} > 0$\\
            We take the block $\beta(d,j)$. 
            \begin{itemize}
                \item $s(\beta(d,j), d) = m s_j$
                \item For all $ x \in C \setminus \{d\}, \text{ we have }s(\beta(d,j), x) = \lambda_{\beta(d,j)} = - s_j +  \sum\limits_{k=1}^{m+1} s_k$
            \end{itemize}
            Consider the vote where $c_i$ is in position $j+1$. By construction, in every block, such a vote exists. Swap the positions of candidate $d$ and $c_i$. Let this block of votes be $\beta^+(d,j)$. 
            \begin{itemize}
                \item $s(\beta^+(d,j), d) = m s_j - s_j + s_{j+1} = m s_j - \delta_j$
                \item $s(\beta^+(d,j), c_i) = \lambda_{\beta(d,j)} - s_{j+1} + s_j = \lambda_{\beta(d,j)} + \delta_{j}$
        
                \item For all $x \in C \setminus \{c_i, d\},$ we have 
                $s(\beta^+(d,j), x) = \lambda_{\beta(d,j)}.$
            \end{itemize}
            
            We add $\eta_{i,j}$ copies of the votes in the block $\beta^+(d,j)$ to the profile $\boldsymbol{Q}$. We add $\eta_{i,j} \lambda_{\beta(d,j)}$ to $\lambda_{\boldsymbol{Q}}$.
                
            \item[Case II.] $\eta_{i,j} < 0$\\
             We take the block $\beta(d,j+1)$. 
             \begin{itemize}
                \item $s(\beta(d,j+1), d) = m s_{j+1}$
                \item For all $ x \in C \setminus \{d\}, \text{ we have }s(\beta(d,j+1), x) = \lambda_{\beta(d,j+1)} = - s_{j+1} +  \sum\limits_{k=1}^{m+1} s_k$
            \end{itemize}
            Consider the vote where $c_i$ is in position $j$. By construction, in every block, such a vote exists. Swap the positions of candidate $d$ and $c_i$. Let this block of votes be $\beta^-(d,j)$. 
            \begin{itemize}
                \item $s(\beta^-(d,j+1), d) = m s_{j+1} + s_j - s_{j+1} = m s_j + \delta_j$
            
                \item $s(\beta^-(d,j+1), c_i) = \lambda_{\beta(d,j+1)} + s_{j+1} - s_j = \lambda_{\beta(d,j+1)} - \delta_{j}$
            
                \item For all $ x \in C \setminus \{c_i, d\},$ we have
                    $s(\beta^-(d,j+1), x) = \lambda_{\beta(d,j+1)}.$
            \end{itemize}
            We add $\eta_{i,j}$ copies of the votes in the block $\beta^-(d,j)$ to the profile $\boldsymbol{Q}$. We add $\eta_{i,j} \lambda_{\beta(d,j+1)}$ to $\lambda_{\boldsymbol{Q}}$.
    \end{description}
    The number of votes produced in this step is $\eta_{i,j} m.$ The time to construct these votes is bounded by $O(\eta_{i,j}, m).$
        \item Observe that the score of $d$ in the block of votes obtained from either of the above cases can be \emph{more} than that of some $c_w \in C$. In particular, the difference between the scores of $d$ and $c_w$ is always strictly less than $\eta_{i,j} m s_1$. Consider the block $\beta(d,m+1)$.The score of $d$ in this block is $0.$  The score of all the candidates in $C$, including $c_w$, in this block is $\lambda_{\beta(d,m+1)} = \left(\sum\limits_{k=1}^{m} s_k \right).$
        To ensure that $c_w$ is never defeated by $d$, we add $\eta_{i,j} m$ copies of the votes in the block $\beta(d,m+1)$ to the profile $\boldsymbol{Q}$. We add $\eta_{i,j}m \lambda_{\beta(d,m+1)}$ to $\lambda_{\boldsymbol{Q}}.$\\
        The number of votes produced in this step is $\eta_{i,j} m^2.$ The time required to construct these votes is bounded by $O(\eta_{i,j}, m^2).$
\end{enumerate}

Now we compute the upper-bound of the number of votes in $\boldsymbol{Q}$ and time taken to construct the profile. For $1 \leq i \leq m$, to set the score of candidate $c_i$ we added $\left( m + m^2 \right) \sum\limits_{j=1}^{m} \eta_{i,j} $ votes in $\boldsymbol{Q}$, and thus a total of $\left( m + m^2 \right) \sum\limits_{i=1}^{m}\sum\limits_{j=1}^{m} \eta_{i,j} $ votes.
Since for each $i$, we have $\sum\limits_{j=1}^{m} \eta_{i,j} \leq O(m^4),$ the total number of votes in $\boldsymbol{Q}$ is bounded above by a polynomial in $m$. The total time required to construct these votes is also bounded above by a polynomial in $m$.
\end{proof}

Let $(\mathcal{X}, \mathcal{Y}, \mathcal{Z}, \mathscr{S})$ be a {\sc 3DM} instance where $\mathscr{S} = \{S_1, \hdots S_t\} \subseteq \mathcal{X} \times \mathcal{Y} \times \mathcal{Z}$ such that $S_i = (x_{i_1}, y_{i_2}, z_{i_3})$, for $1 \leq i \leq t.$
In all the reductions from {\sc 3DM} to \PWPC, 
for each $c_i \in C$, 
the value $R_i$ will be of the form $R_i = \sum\limits_{k=1}^{m} l_k \delta_{k} + \sum\limits_{k=1}^{m+1} h_k s_{k}$ where 
 $\sum\limits_{k=1}^{m} l_k \leq O(m)$ and each $\sum\limits_{k=1}^{m+1}h_k \leq t \leq O(m^3)$.  
Since, for $1 \leq k \leq m$, the score value $s_k = (\delta_k + \hdots + \delta_m),$ and $s_m = 0,$ we have that $R_i = \sum\limits_{k=1}^{m} l_k \delta_{k} + \sum\limits_{k=1}^{m} h_k \left( \sum\limits_{l=k}^{m} \delta_l \right)$. From this, it follows that $R_i = \sum\limits_{j=1}^{m}\eta_{i,j} \delta_j$, where each $\eta_{i,j}$  is the sum of suitable $l_k$'s and $h_k$'s.
\\

In the reductions, we call the candidates corresponding to the elements of the sets in {\sc 3DM}, the \emph{element candidates}. 
We will often need to define some arbitrary total order on a set of candidates with specific properties. For a set $S$, we denote an arbitrary total order on $S$ as $\overrightarrow{S}.$ For $a,b \in S,$ if we want $a$ to be in a higher position than $b$, i.e., $a$ has score less than or equal to $b$, in the total order, we simply state that $b \succ a$ in $\overrightarrow{S}.$

\eat{
Due to space limitations, we give only the first reduction from {\sc 3DM} to the \PWPC~problem with respect to  $2$-approval, and then show how to extend the reduction to the \PWPC~problem with respect to an arbitrary $2$-valued rule other than plurality and veto. We chose the $2$-approval rule because, as mentioned in the Introduction, the \PW~problem with respect to $2$-approval is in \PTIME, when restricted to partitioned partial orders \cite{DBLP:conf/atal/Kenig19} and to truncated partial orders \cite{baumeister2012campaigns}.
}
We start with the reduction from {\sc 3DM} to the \PWPC~problem with respect to  $2$-approval, and then show how to extend the reduction to the \PWPC~problem with respect to an arbitrary $2$-valued rule other than plurality and veto. This is an interesting case because, as mentioned in the Introduction, the \PW~problem with respect to $2$-approval is in \PTIME, when restricted to partitioned partial orders \cite{DBLP:conf/atal/Kenig19} and to truncated partial orders \cite{baumeister2012campaigns}.

\subsection{Hardness of \PWPC \ w.r.t.\ $2$-valued rules}
We first present the reduction of {\sc 3DM} to \PWPC \ w.r.t.\ $2$-approval, and then prove its correctness. 
\begin{reduction}
\label{red:2val}
Let $(\mathcal{X}, \mathcal{Y}, \mathcal{Z}, \mathscr{S})$ be a {\sc 3DM} instance where $\mathscr{S} = \{S_1, \hdots S_t\} \subseteq \mathcal{X} \times \mathcal{Y} \times \mathcal{Z}$ such that $S_i = (x_{i_1}, y_{i_2}, z_{i_3})$, for $1 \leq i \leq t.$ We construct an instance of the \PWPC~ problem as follows.

\begin{enumerate}
    \item The set of candidates is $C = X \cup Y \cup Z \cup \{c, d_1, w\}$ where the sets $X, Y, \text{ and } Z$ comprise of candidates corresponding to the elements of the sets $\mathcal{X}, \mathcal{Y}, \text{ and } \mathcal{Z}$ of the {\sc 3DM} instance.
    
    \item We construct the partial profile $\boldsymbol{P}$ as follows.
    \begin{itemize}
        \item For each $S_i = ( x_{i_1}, y_{i_2}, z_{i_3} )$, let $C_i' = C \setminus ( \{x_{i_1}, y_{i_2}, z_{i_3} \} \cup \{d_1 \})$ and $\overrightarrow{C_i'}$ be such that $c \succ w$, i.e., $w$ is in a position higher than that of $c$.
        \begin{align*}
            p'_i &= x_{i_1} \succ y_{i_2} \succ z_{i_3} > d_1 \succ \overrightarrow{C_i'}\\
            p_i &= x_{i_1} \succ y_{i_2} \succ \overrightarrow{C_i'} 
        \end{align*}
        \item $\boldsymbol{P} = \bigcup_{i=1}^{l} p_i$ is a partial profile where each vote is a partial chain.

         $\boldsymbol{P}' = \bigcup_{i=1}^{l} p'_i$ is a total profile. 
        Moreover, each $p'_i$ extends $p_i.$ Let $s(\boldsymbol{P}', c) = \lambda_{\boldsymbol{P}'} = 0.$ 
        Since $w$ is placed at a position greater $c$ in all the votes of $\boldsymbol{P}',$ we have $s(\boldsymbol{P}', w) = \lambda_{\boldsymbol{P}'}.$ 
    \end{itemize}
    
   \item Consider $C = X \cup Y \cup Z \cup \{ c, d_1 \} \cup \{w\}$. Let $\{w\}$ be the set $D$ required in Lemma \ref{lemmaDM} and $\mathbf{R}$ be as follows.
   \begin{itemize}
       \item  $R_{x_i} = 1 - \left(s(\boldsymbol{P}', x_i) - \lambda_{\boldsymbol{P}'} \right)$, for $1 \leq i \leq q.$
       
       \item $R_{y_i} = 1 - \left(s(\boldsymbol{P}', y_i) - \lambda_{\boldsymbol{P}'} \right)$,
       for $1 \leq i \leq q.$ 
       
       \item $R_{z_i} =  - 1 - \left(s(\boldsymbol{P}', z_i) - \lambda_{\boldsymbol{P}'} \right)$, for $1 \leq i \leq q$. 
       
       \item $R_{d_1} = -q - \left(s(\boldsymbol{P}', d_1) - \lambda_{\boldsymbol{P}'} \right).$
       
       \item $R_{c} = 0.$
   \end{itemize}
    
   \item By Lemma \ref{lemmaDM}, there exist a $\lambda_{\boldsymbol{Q}} \in \mathbb{N}$ and a total profile $\boldsymbol{Q}$ which can be constructed in time polynomial in $m$ such that the scores of the candidates in the profile $\boldsymbol{P}' \cup \boldsymbol{Q}$ are as follows. Let $\lambda_{\boldsymbol{P}'} + \lambda_{\boldsymbol{Q}} = \lambda.$
        
    \begin{itemize}
         \item For all $x \in X$, we have
        $s( \boldsymbol{P}' \cup \boldsymbol{Q}, x) = s(\boldsymbol{P}', x) + s(\boldsymbol{Q},x)$ 
            \begin{align*}
                & = \left(\lambda_{\boldsymbol{P}'} + s(\boldsymbol{P}',x) - \lambda_{\boldsymbol{P}'} \right) + \left(\lambda_{\boldsymbol{Q}} +  R_x \right)  = \lambda + 1.
            \end{align*}  
        
        \item For all $y \in Y$, we have
        $s( \boldsymbol{P}' \cup \boldsymbol{Q}, y) = s(\boldsymbol{P}', y) + s(\boldsymbol{Q},y)$ 
            \begin{align*}
                & = \left(\lambda_{\boldsymbol{P}'} + s(\boldsymbol{P}',y) - \lambda_{\boldsymbol{P}'} \right) + \left(\lambda_{\boldsymbol{Q}} +  R_y \right) = \lambda + 1.
            \end{align*}

        \item For all $z \in Z$, we have
        $s( \boldsymbol{P}' \cup \boldsymbol{Q}, z) = s(\boldsymbol{P}', z) + s(\boldsymbol{Q},z)$ 
            \begin{align*}
                & = \left(\lambda_{\boldsymbol{P}'} + s(\boldsymbol{P}',z) - \lambda_{\boldsymbol{P}'} \right) + \left(\lambda_{\boldsymbol{Q}} +  R_z \right) = \lambda - 1.
            \end{align*}  
    
        \item $s( \boldsymbol{P}' \cup \boldsymbol{Q}, c) = s(\boldsymbol{P}', c) + s(\boldsymbol{Q},c)$ $= \lambda_{\boldsymbol{P}'} + \lambda_{\boldsymbol{Q}} = \lambda.$
        
        \item $s( \boldsymbol{P}' \cup \boldsymbol{Q}, d_1) = s(\boldsymbol{P}', d_1) + s(\boldsymbol{Q},d_1)$ 
        \begin{align*}
                & = \left(\lambda_{\boldsymbol{P}'} + s(\boldsymbol{P}',d_1) - \lambda_{\boldsymbol{P}'} \right) + \left(\lambda_{\boldsymbol{Q}} +  R_{d_1} \right) = \lambda - q.
            \end{align*}   
            
       \item $s( \boldsymbol{P}' \cup \boldsymbol{Q}, w) = s(\boldsymbol{P}', w) + s(\boldsymbol{Q},w)$ $< \lambda_{\boldsymbol{P}'} + \lambda_{\boldsymbol{Q}}  < \lambda.$
       
    \end{itemize}
    
    \item We let $C$, the partial profile $\boldsymbol{V} =  \boldsymbol{P} \cup \boldsymbol{Q}, \text{ and } c$ be the input to the \PWPC~ problem.
\end{enumerate}
\end{reduction}
\begin{proposition}
\label{c_equals_phi_t2}
Let $\boldsymbol{P}, \boldsymbol{P}',$ and $\boldsymbol{Q}$ be the profiles as in the construction above. For all $\overline{\boldsymbol{P}}$ which extend $\boldsymbol{P}$, we have $s(\overline{\boldsymbol{P}} \cup \boldsymbol{Q}, c) = s(\boldsymbol{P}' \cup \boldsymbol{Q}, c) = \lambda.$ 
\end{proposition}
\begin{proof}
Recall that the scoring vector is $(1,1, 0, \hdots, 0).$ By construction, in every completion of the partial chain $p_i$ in $\boldsymbol{P}$, for $1 \leq i \leq m$, candidate $c$ is always in a position greater than two. Thus, for all total profiles $\overline{\boldsymbol{P}}$ which extend the partial profile $\boldsymbol{P},$ the score of $c$ never changes.
\end{proof}
\begin{lemma}
\label{lemma_2val_2}
 \emph{\PWPC}~w.r.t.\ $2$-approval is \emph{\NP}-complete.
\end{lemma}
\begin{proof}
Given a {\sc 3DM} instance $(\mathcal{X}, \mathcal{Y}, \mathcal{Z}, \mathscr{S}),$ we construct a \PWPC \ instance, $(C, \boldsymbol{V} = \boldsymbol{P} \cup \boldsymbol{Q}, c),$ according to Reduction \ref{red:2val}. Let $|C| = m.$ 
\\
First, we prove the $\impliedby$ direction.
Assume that the instance $(C, \boldsymbol{V} = \boldsymbol{P} \cup \boldsymbol{Q}, c)$ of the \PWPC~problem obtained from the reduction is a positive one. Therefore, there exists a total profile $P^* = \bigcup_{i=1}^{l} p^*_i$ such that  

\begin{itemize}
    \item for all $1 \leq i \leq t,$ we have $p^*_i \text{ extends } p_i$;

    \item $c$ is a possible winner and, by Proposition \ref{c_equals_phi_t2}, has score $\lambda.$
\end{itemize}

When we say that a candidate ``gains" or ``loses" points, it is in relation to the complete profile $\boldsymbol{P}'$ in the reduction. 
\begin{enumerate}

    \item  For $1\leq i \leq q$, each element candidate $x_i$ in $X$, , has to lose at least one point. Since a candidate can lose at most one point in any vote, let $p^*_{k_i}$  be the vote in which the element candidate $x_i$  loses a point, where $1 \leq i \leq q$. Let $K = \{ k_i | 1 \leq i \leq q \}$.

    \item  Observe that in all the $q$ votes in $K$, both $d_1$ and an element candidate from $Z$ must be in the top two positions. Without loss of generality, assume that, in these $q$ votes, candidate $d_1$ is in the first position and the element candidate from $Z$ is in the second position. 
    
    \item Therefore, candidate $d_1$ gains a total of $q$ points. Since each $z \in Z$ can gain at most a point, the element candidate of $Z$ in the second position in each of the above $q$ votes must be distinct, i.e., no two votes in $K$ have the same element candidate of $Z$ in the second position.

    \item By construction, candidates $d_1$ and $z$ cannot gain any more points. Since $c$ is a possible winner, it must be the case that each of the $q$ element candidates in $Y$ also lost at least a point each in the $q$ votes in $K$. Therefore, the element candidates of $Y$ in the $q$ votes in $K$ must be distinct.
   
     \item Therefore, the set $\{S_i | i \in K \}$ must form a cover for $\mathcal{X} \cup \mathcal{Y} \cup \mathcal{Z}.$

\end{enumerate}

Now, we prove the other direction. Let $(\mathcal{X}, \mathcal{Y}, \mathcal{Z}, \mathscr{S})$ be a positive instance of {\sc 3DM}. Let $\mathscr{S'} \subseteq \mathscr{S}$ be the cover. Recall that $| \mathscr{S}' | = q.$ We show that $c$ is, indeed, a possible winner in the \PWPC~instance constructed as above. 

\begin{enumerate}

    \item We extend each partial vote $p_i \in \boldsymbol{P}$ as follows.
    \begin{align*}
       p^*_i : d_1 \succ z_{i_3} \succ x_{i_1} \succ y_{i_2} \succ \overrightarrow{C'_i} \text{ if } S_i \in \mathscr{S}' \\
       p^*_i : x_{i_1} \succ y_{i_2} \succ z_{i_3} \succ d_1  \succ \overrightarrow{C_i'} \text{ if } S_i \notin \mathscr{S}'
   \end{align*} 
    Let $\boldsymbol{P}^* = \bigcup_{i=1}^{l} p^*_i$.

\begin{table}[!ht]
\centering
\begin{tabular}{lrrrrrr}  
\toprule
1  & 1 & 0 & 0 & $\hdots$ & 0 & \\
\midrule
       $x_{i_1}$ &  $y_{i_2}$  &  $z_{i_3}$   &  $d_1$ & $C_i$ &    & if $S_i \in \mathscr{S}'$  \\
       
        $d_1$ & $z_{i_3}$ & $x_{i_1}$ &  $y_{i_2}$ & $C_i$&  & if $S_i \notin \mathscr{S}'$  \\
\\
\end{tabular}
\caption{Completions for $2$-approval. }
\label{tab:2_app}
\end{table}

    \item  The following are the scores of the candidates in the profile $\boldsymbol{P}^* \cup \boldsymbol{Q}$. Recall, that $s(\boldsymbol{P}^* \cup \boldsymbol{Q}, c) = \lambda$.
    \begin{itemize}
        \item For all $x \in X,$ we have
        $s(\boldsymbol{P}^* \cup \boldsymbol{Q}, x) = s(\boldsymbol{P}^* , x) + s(\boldsymbol{Q}, x) = s(\boldsymbol{P}' , x) - 1 + s(\boldsymbol{Q}, x) $
        \begin{align*}
            =& \left(\lambda_{\boldsymbol{P}'} + s(\boldsymbol{P}',x) - \lambda_{\boldsymbol{P}'} \right) -1 + \left(\lambda_{\boldsymbol{Q}} +  R_x \right) = \lambda.
        \end{align*}
        
        \item For all $y \in Y,$we have
        $s(\boldsymbol{P}^* \cup \boldsymbol{Q}, y) = s(\boldsymbol{P}^* , y) + s(\boldsymbol{Q}, y) = s(\boldsymbol{P}' , y) - 1 + s(\boldsymbol{Q}, y)$
        \begin{align*}
            =& \left(\lambda_{\boldsymbol{P}'} + s(\boldsymbol{P}',y) - \lambda_{\boldsymbol{P}'} \right) -1 + \left(\lambda_{\boldsymbol{Q}} +  R_y \right) = \lambda.
        \end{align*}
        
        \item For all $z \in Z,$ we have 
        $s(\boldsymbol{P}^* \cup \boldsymbol{Q}, z) = s(\boldsymbol{P}^* , z) + s(\boldsymbol{Q}, z) = s(\boldsymbol{P}' , z) + 1 + s(\boldsymbol{Q}, z) $
        \begin{align*}
            =& \left(\lambda_{\boldsymbol{P}'} + s(\boldsymbol{P}',z) - \lambda_{\boldsymbol{P}'} \right) + 1 + \left(\lambda_{\boldsymbol{Q}} +  R_z \right) = \lambda.
        \end{align*}
        
        \item $s(\boldsymbol{P}^* \cup \boldsymbol{Q}, c) = s(\boldsymbol{P}^* , c) + s(\boldsymbol{Q}, c)  = s(\boldsymbol{P}' , c) + s(\boldsymbol{Q}, c) 
     = \lambda_{\boldsymbol{P}'} + \lambda_{\boldsymbol{Q}} = \lambda.$
     
     \item $s(\boldsymbol{P}^* \cup \boldsymbol{Q}, d_1) = s(\boldsymbol{P}^* , d_1) + s(\boldsymbol{Q}, d_1)  = s(\boldsymbol{P}' , d_1) + q + s(\boldsymbol{Q}, d_1)$
        \begin{align*}
            =& \left(\lambda_{\boldsymbol{P}'} + s(\boldsymbol{P}',d_1)
            - \lambda_{\boldsymbol{P}'} \right) + q + \left(\lambda_{\boldsymbol{Q}} +  R_{d_1} \right) = \lambda.
        \end{align*}

        \item $s( \boldsymbol{P}' \cup \boldsymbol{Q}, w) = s(\boldsymbol{P}', w) + s(\boldsymbol{Q},w)$ $< \lambda_{\boldsymbol{P}'} + \lambda_{\boldsymbol{Q}} < \lambda.$ 
    \end{itemize}
    Therefore, $c$ is a possible winner.
\end{enumerate}
    \vspace{-2pc}
\end{proof}
Next, we generalise the construction to all $2$-valued  rules.
\begin{lemma}
If $r$ is a $2$-valued rule, then
\emph{\PWPC}~w.r.t.\ $r$ is \emph{\NP}-complete.
\end{lemma}

\begin{proof}
\emph{(Outline)}
The set $C$ of candidates is the same as in the reduction for $2$-approval. One can always construct the total profile $\boldsymbol{P}' = \cup_{i=1}^{l} p'_i$, such that in the vote $p'_i$, corresponding to $S_i \in \mathscr{S},$  candidates $x_{i_1}$ and $y_{i_2}$ are in the two highest positions with score value one, and candidates $z_{i_3}$ and $d_1$ are in the two lowest positions with score value zero. The score of each candidate in the profile $\boldsymbol{P}' \cup \boldsymbol{Q}$ is set identical to that in the reduction for $2$-approval. We drop $d_1$ and $z_{i_3}$ in each $p'_i$ to obtain the partial profile $\boldsymbol{P}$. More precisely, if the scoring vector is $(\underbrace{1, \hdots, 1}_{k}, 0, 0, \hdots, 0)$, we do the following.

\begin{itemize}
    \item For each $S_i$, we construct a total order $p'_i$ such that
    candidates $x_{i_1}$ and $y_{i_2}$ are in positions $k-1$ and $k$ respectively, while 
 candidates $z_{i_3}$ and $d$ are in positions $k+1$ and $k+2$.

     $$p'_i = \overrightarrow{ C^1_i} \succ x_{i_1} \succ y_{i_2} \succ z_{i_3} \succ d_1 \succ \overrightarrow{C^2_i},$$
     
    where $C^1_i$ and $C^2_i$ are partitions of $C \setminus \{x_{i_1} , y_{i_2} , z_{i_3} , d_1  \}$ such that $|C_i^1| = r-2$ and $C_i^2 = C \setminus \left( C_i^1 \cup \{x_{i_1} , y_{i_2} , z_{i_3} , d_1  \} \right)$ . Since $2 \leq r \leq (m-2)$, the positions $(k-1), k, (k+1), \text{ and } (k+2)$ are always valid.
    
    \item Construct partial votes $p_i$ by dropping   $d_1$ and $z_{i_3}$ from $p'_i$.
    
    \item For Proposition $\ref{c_equals_phi_t2}$ to hold for all $2$-valued scoring rules, one ensures that  and $|C_i^1| = k-2$ where $k$ is the number of times the largest score value is in the scoring vector.
    
    \item To apply Lemma \ref{lemmaDM}, one ensures that in each $p'_i$  candidate $c$ is an position smaller than that of $w$.
    
    \item The relative scores of each candidate in the profile $\boldsymbol{P}' \cup \boldsymbol{Q}$ is set in a way similar to that in the reduction for $2$-approval (Reduction \ref{red:2val}), and thus the proof of $\impliedby$ direction is similar.
    
    \item For the $\implies$ direction, the partial chains in $\boldsymbol{P}$ are completed as in Table \ref{tab:t_app}. This makes the score of candidate $c$  greater than or equal to the score of all the other candidates. 
\end{itemize}
\vspace{-2pc}
\end{proof}
    \begin{table}[!ht]
\centering
\begin{tabular}{lrrrrrrrrr}  
\toprule
1  & $\hdots$ & 1 & 1 & 0 & 0 & $\hdots$ & 0 &  \\
\midrule
       &  $\overrightarrow{C^1_i}$  & $x_{i_1}$ &  $y_{i_2}$  &  $z_{i_3}$   &  $d_1$ & $\overrightarrow{C^2_i}$ & & if $S_i \in \mathscr{S}'$    \\
       &  $\overrightarrow{C^1_i}$  & $d_1$ & $z_{i_3}$ & $x_{i_1}$ &  $y_{i_2}$ & $\overrightarrow{C^2_i}$ & & if $S_i \notin \mathscr{S}'$    \\
       \\
\end{tabular}
\caption{Completions for a $2$-valued rules.}
\label{tab:t_app}
\end{table}
\subsection{Hardness of \PWPC~ w.r.t.\ $p$-valued rules, where $p \geq 3$}
\label{sub:PWPC_3val}
In this section, we show NP-completeness of $p$-valued positional scoring rules, for $p \geq 3$. Consider a $p$-valued rule, where $p \geq 3$, which has a size $m$ scoring vector with the distinct values $a_1 > a_2 > \hdots > a_p$, we define, for $1 \leq j \leq p$, a function $\ell (m,j)$ which returns the number of times the score value $a_j$ repeats in the scoring vector. Schematically, a scoring vector of a $p$-valued rule, where $p \geq 3$, can be represented as follows.

$$\left( \underbrace{a_1, \hdots, a_1}_{\ell (m,1)}, \underbrace{a_{2}, \hdots, a_{2}}_{\ell(m , 2)}, \hdots, \underbrace{ a_{p}, \hdots, a_p}_{\ell (m, p)} \right)$$

The following proposition follows from the purity of the scoring rules considered in this paper.

\begin{proposition}
\label{prop:p_val_len}
Let $r$ be a $p$-valued scoring rule. For all positive integers $\gamma$, there exists a length $m \leq \gamma p$ such that, in the scoring vector $\boldsymbol{s}_{m}$, there exists $1 \leq u \leq p$ such that $\ell (m, u) = \gamma.$
\end{proposition}

We give the reduction first and then prove its correctness.

\begin{reduction}
\label{red:3val}
Let $\mathcal{I} = (\mathcal{X, Y, Z}, \mathscr{S} )$ be a { \sc 3DM } instance,
with $|\mathcal{X}| = |\mathcal{Y}| = |\mathcal{Z}| = q$. 
Let $r$ be the $p$-valued scoring rule which has scoring vectors with blocks of repeating score values. 
More precisely, in the scoring vector of length $m,$ the score value $a_j$ repeats $\ell (m, j)$ times, for $1 \leq j \leq p.$ Let $\delta_j = a_j - a_{j+1},$ for $1 \leq j < p.$
Let $\gamma = 3q.$ 
By Proposition \ref{prop:p_val_len}, there is a number $m \leq 3q p$ such that in the scoring vector $\boldsymbol{s}_{m}$, there exists $1 \leq u \leq p$ such that the block of repeating score value $a_u$ has length $\ell (m', u) = 3q$. We consider the following three cases.
\begin{enumerate}
    \item[Case 1.]$u = 1$
    \item[Case 2.]$u = p$
    \item[Case 3.]$1 < u < p$
\end{enumerate}
We construct the \PWPC~instance as follows.
\begin{enumerate}
    \item The set of candidates is $C = X \cup Y \cup Z \cup \{c, w\} \cup H$ where $c$ denotes the distinguished candidate, the sets $X, Y, \text{ and } Z$ comprise of candidates corresponding to the elements of the sets $\mathcal{X, Y}$ and $\mathcal{Z}$. The set $H$ consists of dummy candidates such that $|H| = m-3q + 2.$
    
    \item We construct the partial profile $\boldsymbol{P}$ as follows.
    
    \begin{itemize}
        \item For each $S_i = (x_{i_1}, y_{i_2}, z_{i_3} ),$ let  $C'_i = C \setminus (\{x_{i_1}, y_{i_2}, z_{i_3} \} \cup H)$. Let $\overrightarrow{C'_i}$ be such that candidate $c$ is ranked lower than $w$, i.e., we have $c \succ w.$
        
         \paragraph{Case 1.}$u = 1.$  Let $H_1 \subseteq H$ such that $|H_1| = \ell (m, 2) -1$ and $H' = H \setminus H_1$.
         \begin{align*}
                p'_i &= \overrightarrow{C'_i} \succ x_{i_1} \succ y_{i_2} \succ \overrightarrow{H_1} \succ z_{i_3} \succ \overrightarrow{H'}\\
                p_i &= \overrightarrow{C'_i} \succ y_{i_2} \succ \overrightarrow{H_1} \succ z_{i_3} \succ \overrightarrow{H'}
        \end{align*}
        
        \paragraph{Case 2.}$u = p.$  Let $H_1 \subseteq H$ such that $|H_1| = \ell (m, p-1) -1$ and $H' = H \setminus H_1$.
        \begin{align*}
                p'_i &= \overrightarrow{H'} \succ x_{i_1} \succ y_{i_2} \succ \overrightarrow{H_1} \succ z_{i_3} \succ \overrightarrow{C_i'}\\
                p_i &= \overrightarrow{H'} \succ y_{i_2} \succ \overrightarrow{H_1} \succ z_{i_3} \succ \overrightarrow{C_i'}
        \end{align*}
        
        \paragraph{Case 3.}$1 < u < p.$  Let $H_1 \subseteq H$ such that $|H_1| = \sum\limits_{i=1}^{u-2} \ell(m, i) + \ell(m, u-1) -1 $ and $H' = H \setminus H_1$.
        \begin{align*}
                p'_i &= \overrightarrow{H_1} \succ x_{i_1} \succ y_{i_2} \succ \overrightarrow{C'_i} \succ z_{i_3} \succ \overrightarrow{H'}\\
                p_i &= \overrightarrow{H_1} \succ y_{i_2} \succ \overrightarrow{C'_i} \succ z_{i_3} \succ \overrightarrow{H'}
        \end{align*}
        
        \item $\boldsymbol{P} = \bigcup_{i=1}^{l} p_i$ is a partial profile where each vote is a partial chain.

         $\boldsymbol{P}' = \bigcup_{i=1}^{l} p'_i$ is a total profile. 
        Moreover, each $p'_i$ extends $p_i.$ Let $s(\boldsymbol{P}', c) = \lambda_{\boldsymbol{P}'}.$ 
        Since $w$ is placed at a position greater $c$ in all the votes of $\boldsymbol{P}',$ we have $s(\boldsymbol{P}', w) < \lambda_{\boldsymbol{P}'}.$ 
    \end{itemize}

   \item Consider $C = X \cup Y \cup Z \cup \{ c \} \cup H \cup \{w\}$. Let $\{w\}$ be the set $D$ required in Lemma \ref{lemmaDM} and $\mathbf{R}$ be as follows.
   
\paragraph{Case 1.}$u = 1$
    \begin{itemize}
       \item $R_{x_i} = \delta_{1} + \delta_{2} - \left(s(\boldsymbol{P}', x_i) -\lambda_{\boldsymbol{P}'} \right)$, for $1 \leq i \leq q.$ 
       
       \item $R_{y_i} = - \delta_{1} - \left(s(\boldsymbol{P}', y_i) - \lambda_{\boldsymbol{P}'} \right)$, for $1 \leq i \leq q.$
       
       \item $R_{z_i} =  - \delta_{2} - \left(s(\boldsymbol{P}', z_i) - \lambda_{\boldsymbol{P}'} \right)$, for $1 \leq i \leq q.$ 
       
       \item $R_{c} = 0.$
       
       \item $R_{h} = 0 - \left(s(\boldsymbol{P}', h) - \lambda_{\boldsymbol{P}'} \right) $, for all $h \in H$.
   \end{itemize}
   
\paragraph{Case 2.}$u = p$
    \begin{itemize}
       \item $R_{x_i} = \delta_{p-2} + \delta_{p-1} - \left(s(\boldsymbol{P}', x_i) -\lambda_{\boldsymbol{P}'} \right)$, for $1 \leq i \leq q.$
       
       \item  $R_{y_i} = - \delta_{p-2} - \left(s(\boldsymbol{P}', y_i) - \lambda_{\boldsymbol{P}'} \right)$, for $1 \leq i \leq q.$
       
       \item  $R_{z_i} =  - \delta_{p-1} - \left(s(\boldsymbol{P}', z_i) - \lambda_{\boldsymbol{P}'} \right)$, for $1 \leq i \leq q.$
       
       \item $R_{c} = 0.$
       
       \item $R_{h} = 0 - \left(s(\boldsymbol{P}', h) - \lambda_{\boldsymbol{P}'} \right)$, for all $h \in H.$
   \end{itemize}

\paragraph{Case 3.}$1 < u < p$
   \begin{itemize}
       \item $R_{x_i} = \delta_{u-1} + \delta_{u} - \left(s(\boldsymbol{P}', x_i) -\lambda_{\boldsymbol{P}'} \right)$, for $1 \leq i \leq q.$ 
       
       \item $R_{y_i} = - \delta_{u-1} - \left(s(\boldsymbol{P}', y_i) - \lambda_{\boldsymbol{P}'} \right)$, for $1 \leq i \leq q.$
       
       \item $R_{z_i} =  - \delta_{u} - \left(s(\boldsymbol{P}', z_i) - \lambda_{\boldsymbol{P}'} \right)$, for $1 \leq i \leq q.$
       
       \item $R_{c} = 0.$
       
       \item $R_{h} = 0 - \left(s(\boldsymbol{P}', h) - \lambda_{\boldsymbol{P}'} \right) $, for all $h \in H.$
   \end{itemize}

   \item By Lemma \ref{lemmaDM}, there exist a $\lambda_{\boldsymbol{Q}} \in \mathbb{N}$ and a total profile $\boldsymbol{Q}$ which can be constructed in time polynomial in $m'$  such that the scores of the candidates in the profile $\boldsymbol{P}' \cup \boldsymbol{Q}$ are as follows. Let $\lambda_{\boldsymbol{P}'} + \lambda_{\boldsymbol{Q}} = \lambda.$
    
\paragraph{Case 1.}$u = 1$
\begin{itemize}
         \item For all $x \in X$, we have 
        $s( \boldsymbol{P}' \cup \boldsymbol{Q}, x) = s(\boldsymbol{P}', x) + s(\boldsymbol{Q},x)$ 
            \begin{align*}
                & = \left(\lambda_{\boldsymbol{P}'} + s(\boldsymbol{P}',x) - \lambda_{\boldsymbol{P}'} \right) + \left(\lambda_{\boldsymbol{Q}} +  R_x \right)  = \lambda + \delta_{2} +\delta_{1}.
            \end{align*}  
        
        \item For all $y \in Y$, we have
        $s( \boldsymbol{P}' \cup \boldsymbol{Q}, y) = s(\boldsymbol{P}', y) + s(\boldsymbol{Q},y)$ 
            \begin{align*}
                & = \left(\lambda_{\boldsymbol{P}'} + s(\boldsymbol{P}',y) - \lambda_{\boldsymbol{P}'} \right) + \left(\lambda_{\boldsymbol{Q}} +  R_y \right) = \lambda - \delta_{1}.
            \end{align*}

        \item For all $z \in Z$, we have 
        $s( \boldsymbol{P}' \cup \boldsymbol{Q}, z) = s(\boldsymbol{P}', z) + s(\boldsymbol{Q},z)$ 
            \begin{align*}
                & = \left(\lambda_{\boldsymbol{P}'} + s(\boldsymbol{P}',z) - \lambda_{\boldsymbol{P}'} \right) + \left(\lambda_{\boldsymbol{Q}} +  R_z \right) = \lambda - \delta_{2}.
            \end{align*}  
    
    \end{itemize} 
\paragraph{Case 2.}$u = p$
 
 \begin{itemize}
         \item For all $x \in X$, we have 
        $s( \boldsymbol{P}' \cup \boldsymbol{Q}, x) = s(\boldsymbol{P}', x) + s(\boldsymbol{Q},x)$ 
            \begin{align*}
                & = \left(\lambda_{\boldsymbol{P}'} + s(\boldsymbol{P}',x) - \lambda_{\boldsymbol{P}'} \right) + \left(\lambda_{\boldsymbol{Q}} +  R_x \right)  = \lambda + \delta_{p-2} +\delta_{p-1}.
            \end{align*}  
        
        \item For all $y \in Y$, we have 
        $s( \boldsymbol{P}' \cup \boldsymbol{Q}, y) = s(\boldsymbol{P}', y) + s(\boldsymbol{Q},y)$ 
            \begin{align*}
                & = \left(\lambda_{\boldsymbol{P}'} + s(\boldsymbol{P}',y) - \lambda_{\boldsymbol{P}'} \right) + \left(\lambda_{\boldsymbol{Q}} +  R_y \right) = \lambda - \delta_{p-2}.
            \end{align*}

        \item For all $z \in Z$, we have 
        $s( \boldsymbol{P}' \cup \boldsymbol{Q}, z) = s(\boldsymbol{P}', z) + s(\boldsymbol{Q},z)$ 
            \begin{align*}
                & = \left(\lambda_{\boldsymbol{P}'} + s(\boldsymbol{P}',z) - \lambda_{\boldsymbol{P}'} \right) + \left(\lambda_{\boldsymbol{Q}} +  R_z \right) = \lambda - \delta_{p-1}.
            \end{align*}  
 \end{itemize}

\paragraph{Case 3.}
    \begin{itemize}
         \item For all $x \in X$, we have 
        $s( \boldsymbol{P}' \cup \boldsymbol{Q}, x) = s(\boldsymbol{P}', x) + s(\boldsymbol{Q},x)$ 
            \begin{align*}
                & = \left(\lambda_{\boldsymbol{P}'} + s(\boldsymbol{P}',x) - \lambda_{\boldsymbol{P}'} \right) + \left(\lambda_{\boldsymbol{Q}} +  R_x \right)  = \lambda + \delta_{u} +\delta_{u-1}.
            \end{align*}  
        
        \item For all $y \in Y$, we have  
        $s( \boldsymbol{P}' \cup \boldsymbol{Q}, y) = s(\boldsymbol{P}', y) + s(\boldsymbol{Q},y)$ 
            \begin{align*}
                & = \left(\lambda_{\boldsymbol{P}'} + s(\boldsymbol{P}',y) - \lambda_{\boldsymbol{P}'} \right) + \left(\lambda_{\boldsymbol{Q}} +  R_y \right) = \lambda - \delta_{u-1}.
            \end{align*}

        \item For all $z \in Z$, we have 
        $s( \boldsymbol{P}' \cup \boldsymbol{Q}, z) = s(\boldsymbol{P}', z) + s(\boldsymbol{Q},z)$ 
            \begin{align*}
                & = \left(\lambda_{\boldsymbol{P}'} + s(\boldsymbol{P}',z) - \lambda_{\boldsymbol{P}'} \right) + \left(\lambda_{\boldsymbol{Q}} +  R_z \right) = \lambda - \delta_{u}.
            \end{align*}  
       
    \end{itemize}
    
    For all the cases, the score of candidate $c$, the dummy candidates in $H$, and $w$ in the profile $\boldsymbol{P}' \cup \boldsymbol{Q}$ are the same.
    \begin{itemize}
        \item $s( \boldsymbol{P}' \cup \boldsymbol{Q}, c) = s(\boldsymbol{P}', c) + s(\boldsymbol{Q},c)$ $= \lambda_{\boldsymbol{P}'} + \lambda_{\boldsymbol{Q}} = \lambda.$
        
        \item For all $h \in H,$ we have  \\
        $s( \boldsymbol{P}' \cup \boldsymbol{Q}, h) = s(\boldsymbol{P}', h) + s(\boldsymbol{Q},h)
        = \left(\lambda_{\boldsymbol{P}'} + s(\boldsymbol{P}',h) - \lambda_{\boldsymbol{P}'} \right) + \left(\lambda_{\boldsymbol{Q}} +  R_{h} \right) = \lambda.$
            
       \item $s( \boldsymbol{P}' \cup \boldsymbol{Q}, w) = s(\boldsymbol{P}', w) + s(\boldsymbol{Q},w)$ $< \lambda_{\boldsymbol{P}'} + \lambda_{\boldsymbol{Q}}  < \lambda.$
    \end{itemize}
    \item We let $C,$ the profile $\boldsymbol{V} =  \boldsymbol{P} \cup \boldsymbol{Q}, \text{ and } c$ be the input to the \PWPC~problem.
\end{enumerate}
\end{reduction}
\begin{lemma}
\label{lem:p_val_1}
Let $r$ be a $p$-valued scoring rule, where $p \geq 3$. Reduction \ref{red:3val} is a polynomial time reduction of {\sc 3DM } to \emph{\PWPC}~ w.r.t.\ $r$.
\end{lemma}
\begin{proof}
We prove the ``$\implies$" direction first. Let $(\mathcal{X}, \mathcal{Y}, \mathcal{Z}, \mathscr{S})$ be a positive instance of {\sc 3DM}. Let $\mathscr{S'} \subseteq \mathscr{S}$ be the cover. Recall that $|\mathscr{S}' |= q.$ We construct a \PWPC \ instance as in Reduction \ref{red:3val} and show that $c$ is, indeed, a possible winner. 

\begin{enumerate}

    \item We extend each partial vote $p_i \in \boldsymbol{P}$ as follows.
    \paragraph{Case 1.}$u=1$
    \begin{align*}
                p^*_i &= \overrightarrow{C'_i} \succ y_{i_2} \succ \overrightarrow{H_1} \succ z_{i_3} \succ x_{i_1}  \succ \overrightarrow{H'} \text{ if } S_i \in \mathscr{S}'\\
                p^*_i &= \overrightarrow{C'_i} \succ x_{i_1} \succ y_{i_2} \succ \overrightarrow{H_1} \succ z_{i_3} \succ \overrightarrow{H'} \text{ if } S_i \notin \mathscr{S}'
        \end{align*}
    
    \paragraph{Case 2.}$u = p$
    \begin{align*}
                p^*_i &= \overrightarrow{H'} \succ y_{i_2} \succ \overrightarrow{H_1} \succ z_{i_3} \succ x_{i_1} \succ \overrightarrow{C_i'} \text{ if } S_i \in \mathscr{S}'\\
                p^*_i &= \overrightarrow{H'} \succ x_{i_1} \succ y_{i_2} \succ \overrightarrow{H_1} \succ z_{i_3} \succ \overrightarrow{C_i'} \text{ if } S_i \notin \mathscr{S}'
        \end{align*}
   
   \paragraph{Case 3.}$1 < u < p$
   \begin{align*}
                p^*_i &= \overrightarrow{H_1}  \succ y_{i_2} \succ \overrightarrow{C'_i} \succ z_{i_3} \succ x_{i_1} \succ \overrightarrow{H'} \text{ if } S_i \in \mathscr{S}'\\
                p^*_i &= \overrightarrow{H_1} \succ x_{i_1} \succ y_{i_2} \succ \overrightarrow{C'_i} \succ z_{i_3} \succ \overrightarrow{H'} \text{ if } S_i \notin \mathscr{S}'
        \end{align*}
        
    Let $\boldsymbol{P}^* = \bigcup_{i=1}^{l} p^*_i$.
    Note that the sore of $c$ does not change in any extension and is, therefore, $\lambda$.

    \item  Now, we compute the scores of all the candidates in the completed profile to verify that candidate $c$ is, indeed, a possible winner. We show the detailed computation for Case 1. The other two cases are similar.
    
    \paragraph{Case 1.}$u = 1$
    \\
    The following are the scores of the candidates in the profile $\boldsymbol{P}^* \cup \boldsymbol{Q}$. Recall, that $s(\boldsymbol{P}^* \cup \boldsymbol{Q}, c) = \lambda$.
    \begin{itemize}
        \item For all $x \in X,$ we have 
        $s(\boldsymbol{P}^* \cup \boldsymbol{Q}, x) = s(\boldsymbol{P}^* , x) + s(\boldsymbol{Q}, x) = s(\boldsymbol{P}' , x) - (\delta_{2} + \delta_{1}) + s(\boldsymbol{Q}, x) $
        \begin{align*}
            =& \left(\lambda_{\boldsymbol{P}'} + s(\boldsymbol{P}',x) - \lambda_{\boldsymbol{P}'} \right) - (\delta_{2} + \delta_{1}) + \left(\lambda_{\boldsymbol{Q}} +  R_x \right) = \lambda.
        \end{align*}
        
        \item For all $y \in Y,$ we have 
        $s(\boldsymbol{P}^* \cup \boldsymbol{Q}, y) = s(\boldsymbol{P}^* , y) + s(\boldsymbol{Q}, y) = s(\boldsymbol{P}' , y) + \delta_{p-2} + s(\boldsymbol{Q}, y).$
        \begin{align*}
            =& \left(\lambda_{\boldsymbol{P}'} + s(\boldsymbol{P}',y) - \lambda_{\boldsymbol{P}'} \right) + \delta_{1} + \left(\lambda_{\boldsymbol{Q}} +  R_y \right) = \lambda.
        \end{align*}
        
        \item For all $z \in Z,$ we have 
        $s(\boldsymbol{P}^* \cup \boldsymbol{Q}, z) = s(\boldsymbol{P}^* , z) + s(\boldsymbol{Q}, z) = s(\boldsymbol{P}' , z) + \delta_{p-1} + s(\boldsymbol{Q}, z) .$
        \begin{align*}
            =& \left(\lambda_{\boldsymbol{P}'} + s(\boldsymbol{P}',z) - \lambda_{\boldsymbol{P}'} \right) + \delta_{2} + \left(\lambda_{\boldsymbol{Q}} +  R_z \right) = \lambda.
        \end{align*}
        
        \item $s(\boldsymbol{P}^* \cup \boldsymbol{Q}, c) = s(\boldsymbol{P}^* , c) + s(\boldsymbol{Q}, c)  = s(\boldsymbol{P}' , c) + s(\boldsymbol{Q}, c) 
     = \lambda_{\boldsymbol{P}'} + \lambda_{\boldsymbol{Q}} = \lambda.$
     
     \item For all $c' \in H,$ we have  $s(\boldsymbol{P}^* \cup \boldsymbol{Q}, c') = s(\boldsymbol{P}^* , c') + s(\boldsymbol{Q}, c')  = s(\boldsymbol{P}' , c') + 0 + s(\boldsymbol{Q}, c')$
        \begin{align*}
            =& \left(\lambda_{\boldsymbol{P}'} + s(\boldsymbol{P}',c')
            - \lambda_{\boldsymbol{P}'} \right) + 0 + \left(\lambda_{\boldsymbol{Q}} +  R_{c'} \right) = \lambda.
        \end{align*}

        \item $s( \boldsymbol{P}' \cup \boldsymbol{Q}, w) = s(\boldsymbol{P}', w) + s(\boldsymbol{Q},w)$ $< \lambda_{\boldsymbol{P}'} + \lambda_{\boldsymbol{Q}} < \lambda.$ 
    \end{itemize}
    Therefore, $c$ is a possible winner.
    
\end{enumerate}
In the other direction, we prove the correctness of the reduction for Case 1 in full detail. The other two cases are similar.
\\
Given a {\sc 3DM} instance $(\mathcal{X}, \mathcal{Y}, \mathcal{Z}, \mathscr{S}),$ we construct a \PWPC~ instance, $(C, \boldsymbol{V} = \boldsymbol{P} \cup \boldsymbol{Q}, c),$ according to the above reduction.
Assume that the \PWPC~ instance  $(C, \boldsymbol{V} = \boldsymbol{P} \cup \boldsymbol{Q}, c)$ is a positive one. Therefore, there exists a total profile $P^* = \bigcup_{i=1}^{l} p^*_i$ such that  

\begin{itemize}
    \item for all $1 \leq i \leq t,$ the vote $p^*_i \text{ extends } p_i$;
    \item $c$ is a possible winner. Moreover, no matter how the partial orders are completed, the score of $c$ is $\lambda.$ 
\end{itemize}

When we say that a candidate ``gains" or ``loses" points, it is in relation to the complete profile $\boldsymbol{P}'$ in the reduction. 
\begin{enumerate}
    
    \item  For $1\leq i\leq q$, each element candidate $x_i$ in $X$, has to lose at least $(\delta_{2} + \delta_{1})$ points. Therefore, it has to be in a position greater than $\ell (m, 1)$ in at least one vote. Assume, for now, that each $x_i$  loses at least $(\delta_{2} + \delta_{1})$ points in one vote, i.e., it is in a position greater than or equal to $\ell (m, 1) + \ell (m, 2)$.  
    Let these $q$ votes be $p_{k_1}, \hdots, p_{k_q}$ where $1 \leq k_i \leq t$ and $K = \{ k_i | 1 \leq i \leq q \}$.
    
    \item For each $i \in K$, in the completion $p^*_{i}$ , the element  candidate $x_{i_1}$  loses the points (and, therefore, is in position greater than $\ell (m, 1) + \ell (m, 2)$), candidates $z_{i_3}$ and $y_{i_2}$ gain $\delta_{1}$ and $\delta_{2}$ points respectively.
    
    \item By construction, each element candidate of $Y$ can gain at most $\delta_{1}$ points, and each element candidate of $Z$ can gain at most $\delta_{2}$ points. Moreover, there are no votes where these element candidates can lose points. Therefore, the element candidates of $Y$ and the element candidates of $Z$, which gain points in the $q$ votes in $K$ must be distinct. 
    
    \item We had assumed that each element candidates of $X$ loses at least $(\delta_{2} + \delta_{1})$ points in one vote. Observe that whenever $x \in X$ is in a position greater than $\ell (m', 1),$ an element candidate of $Y$ gains the maximum points it can without defeating $c$, i.e., $\delta_1$ points. Since there are $q$ element candidates in $X$ and $q$ element candidates in $Y$, every time an element candidate of $Y$ gains $\delta_1$ points, an element candidate of $X$ must lose at least $(\delta_1 + \delta_2)$ points. 
    
    \item The remaining partial votes in $\boldsymbol{P}$ ($p_i$ for $1 \leq i \leq t$ and $i \notin K$), must have the same completion as in $\boldsymbol{P}'.$ 
   
     \item Therefore, the set $\{S_i | i \in K \}$ must form a cover for $\mathcal{X} \cup \mathcal{Y} \cup \mathcal{Z}.$
\end{enumerate}
\vspace{-2pc}
\end{proof}
%
\subsection{Hardness of \PWPC~ w.r.t.\ unbounded rules}
In this section, we focus on unbounded rules. The Borda count is an example of such rules. As noted earlier, unbounded scoring rules may have score values which repeat in blocks. Moreover, unlike Borda count, the score values can be non-uniformly decreasing. Recall, that for a scoring vector of length $m$, with $m'$ distinct score values, the function $\ell (m, j)$ returns the number of times the distinct score value $a_j$ repeats in a block, for $1 \leq j \leq m'$. Schematically, such a scoring vector can be represented as 

$$\left( \underbrace{a_1, \hdots, a_1}_{\ell (m,1)}, \underbrace{a_{2}, \hdots, a_{2}}_{\ell(m , 2)}, \hdots, \underbrace{ a_{m'}, \hdots, a_{m'}}_{\ell (m, m')} \right).$$

Now, we prove a fundamental property of scoring vectors of all unbounded rules.

\begin{proposition}
\label{prop:ub_val_len}
Let $r$ be a positional scoring rule and let $\gamma$ and $\beta$ be two positive integers greater than $1$.  Consider the scoring vector $\boldsymbol{s}_m$ of $r$ with length $m = \gamma \beta$. Then either $\boldsymbol{s}_m$ contains at least $\beta$ distinct values or there exists $1 \leq u \leq \gamma \beta $ such that $\ell(\gamma \beta, u) \geq  \gamma$.
\end{proposition}
\begin{proof}
If the scoring vector of length $\gamma \beta$ contains at least $\beta$ distinct values we are done. Assume it contains fewer than $\beta$ distinct values. But, by the monotonicity of the rules, if two score values are the same, then they must be in the same block. So, we must have a block in which the same score value repeats more than $\gamma$ times, else the total length would be less than $\gamma \beta$, i.e., there exists $u \leq \gamma \beta$ such that  $\ell(\gamma \beta, u) \geq  \gamma$.
\end{proof}

\begin{reduction}
\label{red:PWPC_ub}
Let $(\mathcal{X}, \mathcal{Y}, \mathcal{Z}, \mathscr{S})$ be a 3DM instance where and $\mathscr{S} = \{S_1,  \hdots S_t\} \subseteq \mathcal{X} \times \mathcal{Y} \times \mathcal{Z}$ such that $S_i = (x_{i_1}, y_{i_2}, z_{i_3})$, for $1 \leq i \leq t.$ 
Let $\boldsymbol{s}_{m}$ be the scoring vector of length $m = (3q+4)(3q)$. By Proposition \ref{prop:ub_val_len}, we need to consider the following two cases.
\begin{itemize}
    \item Case 1. There exists a $u$ such that $\ell(m, u) = 3q.$
    \item Case 2. There are $m' = 3q + 4$ distinct values.
\end{itemize}
For Case 1, the reduction mimics Reduction \ref{red:3val} to create a \PWPC~ instance. For Case 2, the reduction proceeds as follows.

Let $a_1 > a_2 > \hdots > a_{m'}$ be the $m'$ distinct values. We define $\boldsymbol{\delta} = (\delta_1, \hdots, \delta_{m'-1})$ where, $\delta_j = a_j - a_{j+1}$, for $1 \leq j < m'$. 

We construct the following instance of the \PWPC~problem.
\begin{enumerate}
    \item The set of candidates is $C = X \cup Y \cup Z \cup  \{c, g, d,  w\} \cup H$ where $X$, $Y$, and $Z$ contains candidates corresponding to the elements in $\mathcal{X}, \mathcal{Y}, \text{ and } \mathcal{Z}$ respectively. These candidates are called \emph{elements candidates}. The set $H$ contains dummy candidates such that $|H| = m-m'.$
    
    \item We construct the partial profile $\boldsymbol{P}$ as follows.
    \begin{itemize}
        \item Let the set $H$ be partitioned into $H_1, \hdots, H_{m'}$, such that $|H_j| = \ell (m, j) -1$, for $1 \leq j \leq m'$. For each $S_i = (x_{i_1}, y_{i_2}, z_{i_3} )$, let $C_i' = C \setminus \left( \{x_{i_1}, y_{i_2}, z_{i_3} \} \cup \{g, d\} \cup \bigcup\limits_{j = m'-4}^{m'} H_j \right)$ and $\overrightarrow{C_i'}$ be such that the dummy candidates in $H_j$ are in a position with score value $a_j$, for $1 \leq j \leq m'-3$ and candidate $c$ is ranked lower than candidate $w$.
        \begin{align*}
            p'_i &= \overrightarrow{C_i'} \succ g \succ \overrightarrow{H}_{m'-4} \succ d \succ \overrightarrow{H}_{m'-3} \succ x_{i_1} \succ \overrightarrow{H}_{m'-2} \succ y_{i_2} \succ \overrightarrow{H}_{m'-1} \succ z_{i_3} \succ \overrightarrow{H}_{m'} \\
            p_i &= \overrightarrow{C_i'} \succ \overrightarrow{H}_{m'-4} \succ d \succ \overrightarrow{H}_{m'-3} \succ x_{i_1} \succ \overrightarrow{H}_{m'-2} \succ y_{i_2} \succ \overrightarrow{H}_{m'-1} \succ z_{i_3} \succ \overrightarrow{H}_{m'}
        \end{align*}
        \item $\boldsymbol{P} = \bigcup_{i=1}^{l} p_i$ is partial profile where each vote is a \linearpartial~where only one candidate ($g$) has been dropped.
        
        \item $\boldsymbol{P}' = \bigcup_{i=1}^{l} p_i'$ is a total profile. Moreover, each $p'_i$ extends $p_i.$ Let $s(\boldsymbol{P}', c) = \lambda_{\boldsymbol{P}'}.$ Observe that $s(\boldsymbol{P}', w) < \lambda_{\boldsymbol{P}'}$ since $w$ is in a position greater $c$ in all $\overrightarrow{C_i'}$, for $1 \leq i \leq t.$ 
    \end{itemize}
    
    \item Consider $C = X \cup Y \cup Z \cup \{c, g,  d\} \cup \{w\}$. Let $\{w\}$ be the set $D$ required in Lemma \ref{lemmaDM} and $\mathbf{R}$ be as follows
    \begin{itemize}
        \item  $R_{x_i} = -\delta_{m'-3}- \left(s(\boldsymbol{P}', x_i) - \lambda_{\boldsymbol{P}'} \right)$, for $1 \leq i \leq q$.
        
        \item $R_{y_i} = -\delta_{m'-2}- \left(s(\boldsymbol{P}', y_i) - \lambda_{\boldsymbol{P}'} \right)$, for $1 \leq i \leq q$.
        
        \item  $R_{z_i} = -\delta_{m'-1}- \left(s(\boldsymbol{P}', z_i) - \lambda_{\boldsymbol{P}'} \right)$, for $1 \leq i \leq q$.
        \item $R_{c} = 0.$
        
         \item $R_{g} = q \left( \sum\limits_{j = 1}^{4}\delta_{m'- j} \right) - \left(s(\boldsymbol{P}',g) - \lambda_{\boldsymbol{P}'}\right).$
        
        \item $R_d = - q(\delta_{m'-4}) - (s(\boldsymbol{P}' ,d) - \lambda_{\boldsymbol{P}'}).$
        
        \item $R_{h} = 0 - s(\boldsymbol{P}', h) - \lambda_{\boldsymbol{P}'} $, for all $h \in H.$
       
    \end{itemize}
    \item By Lemma \ref{lemmaDM}, there exist a $\lambda_{\boldsymbol{Q}} \in \mathbb{N}$ and a total profile $\boldsymbol{Q}$ which can be constructed in time polynomial in $m'$ such that the scores of the candidates in the profile $ \boldsymbol{P}' \cup \boldsymbol{Q}$ are as follows. Let $\lambda_{\boldsymbol{P}'} + \lambda_{\boldsymbol{Q}} = \lambda.$
    \begin{itemize}
        \item For all $x \in X $, we have 
        $s( \boldsymbol{P}' \cup \boldsymbol{Q}, x) = s(\boldsymbol{P}', x) + s(\boldsymbol{Q},x)$ 
            \begin{align*}
                & = \left(\lambda_{\boldsymbol{P}'} + s(\boldsymbol{P}',x) - \lambda_{\boldsymbol{P}'} \right) + \left(\lambda_{\boldsymbol{Q}} +  R_x \right) 
                 = \lambda - \delta_{m'-3}.
            \end{align*} 
            
        \item For all $y \in Y $, we have  
        $s( \boldsymbol{P}' \cup \boldsymbol{Q}, y) = s(\boldsymbol{P}', y) + s(\boldsymbol{Q},y)$ 
            \begin{align*}
                & = \left(\lambda_{\boldsymbol{P}'} + s(\boldsymbol{P}',y) - \lambda_{\boldsymbol{P}'} \right) + \left(\lambda_{\boldsymbol{Q}} +  R_y \right)  = \lambda - \delta_{m'-2}.
            \end{align*} 
            
        \item For all $z \in Z $,  we have 
        $s( \boldsymbol{P}' \cup \boldsymbol{Q}, z) = s(\boldsymbol{P}', z) + s(\boldsymbol{Q},z)$ 
            \begin{align*}
                & = \left(\lambda_{\boldsymbol{P}'} + s(\boldsymbol{P}',z) - \lambda_{\boldsymbol{P}'} \right) + \left(\lambda_{\boldsymbol{Q}} +  R_z \right)  = \lambda - \delta_{m'-1}.
            \end{align*} 
    
        \item $s( \boldsymbol{P}' \cup \boldsymbol{Q}, c) = s(\boldsymbol{P}', c) + s(\boldsymbol{Q},c) = \lambda_{\boldsymbol{P}'} + \lambda_{\boldsymbol{Q}} = \lambda.$

        \item $s( \boldsymbol{P}' \cup \boldsymbol{Q}, g) = s(\boldsymbol{P}', g) + s(\boldsymbol{Q},g)
        = \left(\lambda_{\boldsymbol{P}'} + s(\boldsymbol{P}',g) - \lambda_{\boldsymbol{P}'} \right) + \left(\lambda_{\boldsymbol{Q}} +  R_g \right) 
        = \lambda + q \left( \sum\limits_{j = 1}^{4}\delta_{m'-j} \right).$ 
            
         \item $s( \boldsymbol{P}' \cup \boldsymbol{Q}, d) = s(\boldsymbol{P}', d) + s(\boldsymbol{Q},d) = \left(\lambda_{\boldsymbol{P}'} + s(\boldsymbol{P}',d) - \lambda_{\boldsymbol{P}'} \right) + \left(\lambda_{\boldsymbol{Q}} +  R_d \right) = \lambda - q(\delta_{m'-4}).$ 
      
        \item For all $h \in H,$  we have   \\
        $s( \boldsymbol{P}' \cup \boldsymbol{Q}, h) = s(\boldsymbol{P}', h) + s(\boldsymbol{Q},h) = \left(\lambda_{\boldsymbol{P}'} + s(\boldsymbol{P}',h) - \lambda_{\boldsymbol{P}'} \right) + \left(\lambda_{\boldsymbol{Q}} +  R_h \right)  = \lambda.$
        
       \item $s( \boldsymbol{P}' \cup \boldsymbol{Q}, w) = s(\boldsymbol{P}', w) + s(\boldsymbol{Q},w) < \lambda_{\boldsymbol{P}'} + \lambda_{\boldsymbol{Q}} 
                 < \lambda.$ 
                 
    \end{itemize}

    \item We let $C,$ the profile $\boldsymbol{V} =  \boldsymbol{P} \cup \boldsymbol{Q}, \text{ and } c$ be the input to the \PWPC~ problem.
\end{enumerate} 
\end{reduction}
The following propositions follow quite naturally from the construction of the partial profile in the above reduction (Reduction \ref{red:PWPC_ub}).
\begin{proposition}
\label{prop:c_less_equal_phi_SDR}
Let $\boldsymbol{P}, \boldsymbol{P}',$ and $\boldsymbol{Q}$ be the profiles as constructed in Reduction \ref{red:PWPC_ub}. For all total profiles $\overline{\boldsymbol{P}}$ that extend $\boldsymbol{P}$, we have $s(\overline{\boldsymbol{P}} \cup \boldsymbol{Q}, c) \leq s(\boldsymbol{P}' \cup \boldsymbol{Q}, c) \leq \lambda.$ 
\end{proposition}

\begin{proposition}
\label{prop:c_equals_phi_SDR}
Let $\boldsymbol{P}, \boldsymbol{P}',$ and $\boldsymbol{Q}$ be the profiles as constructed in Reduction \ref{red:PWPC_ub}. For all total profiles $\overline{\boldsymbol{P}}$ that extend $\boldsymbol{P}$, if $c$ is a possible winner in $\overline{\boldsymbol{P}} \cup \boldsymbol{Q}$ then $s(\overline{\boldsymbol{P}} \cup \boldsymbol{Q}, c) = s(\boldsymbol{P} \cup \boldsymbol{Q}, c) = \lambda.$ 
\end{proposition}
Note that the converse of the above proposition is not true.
\begin{lemma}
\label{lem:ub_val}
 Let $r$ be an unbounded scoring rule. Reduction \ref{red:PWPC_ub} is a polynomial time reduction of {\sc 3DM } to \emph{\PWPC}~w.r.t.\ $r$.
\end{lemma}
\begin{proof}
    Given a 3DM instance $\mathcal{I} = (\mathcal{X}, \mathcal{Y}, \mathcal{Z}, \mathscr{S}),$ we construct a PW instance, $(C, V = \boldsymbol{P} \cup \boldsymbol{Q}, c),$ according to Reduction \ref{red:PWPC_ub}. Note that there are two cases in the reduction. In both the cases, the \PW~instance is polynomial in $|\mathcal{I}|$.   For Case 1, correctness follows from Lemma \ref{lem:p_val_1}. Here, we consider Case 2.
    \\
    First, we prove the `$\impliedby'$ direction.
   Assume that the \PWPC~instance $(C, V = \boldsymbol{P} \cup \boldsymbol{Q}, c)$ obtained above is a positive one. Therefore, there exists a total profile $\boldsymbol{P}^* = \bigcup_{i=1}^{l} p^*_i$ such that  
    \begin{itemize}
        \item for all $ 1 \leq i \leq t,$ the vote $p^*_i \text{ extends } p_i;$
        \item $c$ is a possible winner and, by Proposition \ref{prop:c_equals_phi_SDR}, has score $\lambda.$
    \end{itemize}
In the following, when one says that a candidate "gains" or "loses" points, it is in relation to the complete profile $\boldsymbol{P}'$ in the reduction.
    \begin{enumerate}
        \item Candidate $g$ must lose at least $q  \sum\limits_{j=1}^{4}\delta_{m'-j} $ points for $c$ to be a possible winner. Therefore, it must be in a position greater than $\sum\limits_{j = 1}^{m'-4} \ell(m,j)$ at least $q$ times.
        
        \item Whenever $g$ is in a position greater than $m- \sum\limits_{j = m'-3}^{m'} \ell(m, j)$, candidate $d$ gains $\delta_{m'-4}.$  Since $d$ cannot gain more than $q(\delta_{m'-4})$ points, there are at most $q$ votes where $g$ is in position greater than $m- \sum\limits_{j =m'-3}^{m'} \ell(m, j)$. Let these votes be $ p^*_{k_1}, \hdots, p^*_{k_q}$ where each $1 \leq k_j \leq t$ and $K = \{ k_j | 1 \leq j  \leq q \}.$
        
        \item  Note that candidate $g$ has to lose at least $q(\sum\limits_{j = 1}^{4} \delta_{m'-j})$ points in these $q$ votes. This is possible if and only if it is in position greater than $\sum\limits_{j = 1}^{m'-1} \ell (m, j)$. 
        Furthermore, whenever $g$ is in position greater than $\sum\limits_{j = 1}^{m'-1} \ell (m, j)$ in a vote $p^*_i,$ candidate $x_{i_1}$ gains $\delta_{m'-3}$ points, candidate $y_{i_2}$ gains $ \delta_{m'-2}$ points, and candidate $_{i_3}$ gain $\delta_{m'-1}$ points, for $i \in K$.
        
        \item Since $|X| = |Y| = |Z| = q$, and each $x \in X$ can gain at most $\delta_{m-3}$ points, each $y \in Y$ can gain at most $\delta_{m'-2}$ points, and each $z \in Z$ can gain at most $\delta_{m'-1}$ points,  it must be the case that the element candidates of $Y$ and $Z$ which gained points in the $q$ votes in $K$ are distinct. 
        
        \item Since no other candidate can gain any more points, the remaining partial votes in $\boldsymbol{P}$ ($p_i$ for $1 \leq i \leq t$ and $i \notin K$), must have the same completion as in $\boldsymbol{P}'.$ 
   
         \item Therefore, the set $\{S_i | i \in K\}$ must form a cover for $\mathcal{X} \cup \mathcal{Y} \cup \mathcal{Z}.$
    \end{enumerate}
    For the `$\implies$' direction, let $(\mathcal{X}, \mathcal{Y}, \mathcal{Z}, \mathscr{S}),$ is a positive instance. We show that in the PW instance, $(C, V = \boldsymbol{P} \cup \boldsymbol{Q}, c),$ constructed in the reduction, $c$ is, indeed, a possible winner. By hypothesis, there is a $\mathscr{S}' \subseteq \mathscr{S}$ and $| \mathscr{S}' | = q.$
    \begin{enumerate}
        \item Complete each vote  $p_i \in \boldsymbol{P}$ to $p^*_i$.
        \begin{itemize}
            \item $p^*_i = \overrightarrow{C_i'} \succ \overrightarrow{H}_{m'-4} \succ d \succ \overrightarrow{H}_{m'-3} \succ x_{i_1} \succ \overrightarrow{H}_{m'-2} \succ y_{i_2} \succ \overrightarrow{H}_{m'-1} \succ z_{i_3} \succ g  \succ \overrightarrow{H}_{m'}$ if $S_i \in \mathscr{S}'$
            
            \item $p^*_i = \overrightarrow{C_i'} \succ g \succ \overrightarrow{H}_{m'-4} \succ d \succ \overrightarrow{H}_{m'-3} \succ x_{i_1} \succ \overrightarrow{H}_{m'-2} \succ y_{i_2} \succ \overrightarrow{H}_{m'-1} \succ z_{i_3} \succ \overrightarrow{H}_{m'}$ if $S_i \notin \mathscr{S}'$
        \end{itemize}
        
        Let $\boldsymbol{P}^* = \bigcup_{i=1}^{l} p^*_i.$
        Note that the score of candidate $c$ does not change in these votes and is, therefore, $\lambda.$
        
        \item We compute the scores of each candidate in $\boldsymbol{P}^* \cup \boldsymbol{Q}.$
        
        \begin{itemize}
            \item For all $x \in X,$ we have $
            s( \boldsymbol{P}^* \cup \boldsymbol{Q}, x) 
            = s(\boldsymbol{P}^*, x) + s(\boldsymbol{Q},x) 
            =  s(\boldsymbol{P}', x) + \delta_{m'-3}   + s(\boldsymbol{Q},x)$
            \begin{align*}
            =  \lambda_{\boldsymbol{P}'} + s(\boldsymbol{P}',x)  - \lambda_{\boldsymbol{P}'}  + \delta_{m'-3} + \lambda_{\boldsymbol{Q}} +  R_x  = \lambda.
            \end{align*}

            \item For all $y \in Y, $ we have $ 
            s( \boldsymbol{P}^* \cup \boldsymbol{Q}, y) 
            = s(\boldsymbol{P}^*, y) + s(\boldsymbol{Q},y)
            = s(\boldsymbol{P}', y)  + \delta_{m'-2} + s(\boldsymbol{Q},y) $
            \begin{align*}
                = \lambda_{\boldsymbol{P}'} + s(\boldsymbol{P}',y) - \lambda_{\boldsymbol{P}'}   + \delta_{m'-2} + \left(\lambda_{\boldsymbol{Q}} +  R_y \right) = \lambda. 
            \end{align*}
            
            \item  For all $z \in Z , $ we have $
            s( \boldsymbol{P}^* \cup \boldsymbol{Q}, z) 
            = s(\boldsymbol{P}^*, z) + s(\boldsymbol{Q},z)
            = s(\boldsymbol{P}', z)  + \delta_{m'-1} + s(\boldsymbol{Q},z)$ 
            \begin{align*}
                = \left(\lambda_{\boldsymbol{P}'} + s(\boldsymbol{P}',z) - \lambda_{\boldsymbol{P}'} \right) + \delta_{m'-1} + \left(\lambda_{\boldsymbol{Q}} +  R_z \right) = \lambda.
            \end{align*}
                
            \item $s( \boldsymbol{P}^* \cup \boldsymbol{Q}, c) = s(\boldsymbol{P}^*, c) + s(\boldsymbol{Q},c)
                = s(\boldsymbol{P}', c) + s(\boldsymbol{Q},c) = \lambda_{\boldsymbol{P}'} + \lambda_{\boldsymbol{Q}} = \lambda.$
                
            \item $s( \boldsymbol{P}^* \cup \boldsymbol{Q}, g) = s(\boldsymbol{P}^*, g) + s(\boldsymbol{Q},g)
                = s(\boldsymbol{P}', g)  - q \sum\limits_{j=1}^{4} \delta_{m'-j} + s(\boldsymbol{Q},g) $
            \begin{align*}
            = \left(\lambda_{\boldsymbol{P}'} +
            s(\boldsymbol{P}',g) - \lambda_{\boldsymbol{P}'} \right)  - q \sum_{j = 1}^{4} \delta_{m'-j} 
            +\left(\lambda_{\boldsymbol{Q}} +  R_d \right)
            = \lambda.
            \end{align*}
                
            \item $s( \boldsymbol{P}^* \cup \boldsymbol{Q}, d) = s(\boldsymbol{P}^*, d) + s(\boldsymbol{Q},d)= s(\boldsymbol{P}', d)  + (\delta_{m'-4})q + s(\boldsymbol{Q},d)$
            \begin{align*}
            = \left(\lambda_{\boldsymbol{P}'} + s(\boldsymbol{P}',d) - \lambda_{\boldsymbol{P}'} \right) + (\delta_{m'-4})q  +\left(\lambda_{\boldsymbol{Q}} +  R_d \right)  = \lambda. 
            \end{align*}
        
       \item The score of the candidates in $H$ remain unchanged, i.e., for all $h \in H,$ $s( \boldsymbol{P}^* \cup \boldsymbol{Q}, h) = s( \boldsymbol{P}' \cup \boldsymbol{Q}, h) = \lambda.$
                 
        \end{itemize}
    Therefore, $c$ is a possible winner.
    \end{enumerate}
    \vspace{-2pc}
\end{proof}

\section{Beyond Partial Chains} \label{sec:other}
\emph{Partitioned preferences} \cite{DBLP:conf/atal/Kenig19}  and \emph{truncated preferences} \cite{baumeister2012campaigns} are two restricted types of partial orders that have received  attention in the literature.

\begin{definition}
Let $\succ$ be a partial order on a set $C$.
\begin{itemize}
    \item We say that $\succ $ is a \emph{partitioned preference}
    if $C$ can be partitioned into disjoint subsets $A_1, . . . ,A_q$ such that:
    
    \emph{(a) for all $i < j \leq  q$, if $c \in A_i$ and $c' \in  A_j$ then $c \succ c'$;}
    
    \emph{(b)  for each $i < q$, the elements in $A_i$ are incomparable under $\succ$ 
        (i.e., $a \nsucc b$ and $b \nsucc a$, for every $a, b \in A_i$).}
        
\item We say that $\succ$ is a \emph{doubly-truncated ballot} if there is a permutation $\pi$ over $\{1, \hdots, |C| \}$ and natural numbers $t$, $b$ such that $\succ$ is of the form  
    $c_{\pi(1)} \succ \hdots \succ c_{\pi(t)} \succ \{c_{\pi(t+1)}, \hdots, c_{\pi(m-b)} \}
    \succ c_{\pi(m-b+1)} \succ \hdots \succ c_{\pi(m)}$.
    \item A doubly-truncated ballot is called
\emph{top-truncated} if $b=0$; it is called   
 bottom-truncated if $t=0$.
 \end{itemize}
\end{definition}

Note that doubly-truncated ballots are a special case of partitioned preferences; thus, so are top-truncated ballots and bottom-truncated ballots. 

We write \PWPV~to denote the 
the restriction of the \PW~problem to partial profiles consisting of partitioned preferences.
Similarly, we write \PWDV, \PWTV, and \PWBV~for the restriction of the \PW~problem to partial profiles consisting of, respectively, doubly-truncated, top-truncated, and bottom-truncated preferences.

The \PWPV \ problem has been studied in \cite{DBLP:conf/atal/Kenig19}.
Before summarising the main results of that paper, we introduce the concept of a differentiating rule and notation for a family of rules.  

\begin{definition}
A scoring rule $r$ is \emph{differentiating} if there exists a $n_0 \in \mathbb{N}_0$ such that for all $m > n_0$, the scoring vector $\boldsymbol{s}_{m}$ contains two positions $i \text{ and } j$, where $1 \leq i < j < m,$ such that $(s_i - s_{i+1}) > (s_j - s_{j+1}).$ We say $r$ is \emph{non-differentiating} if it is not differentiating.
\end{definition}

Let $f$ and $l$ be two positive integers ($f$ for ``first" and $l$ for ``last"). We write $R(f,l)$ to denote the $3$-valued rule with scoring vectors
${\bf s}_m= (\underbrace{2, \hdots, 2}_{f},\underbrace{1, \hdots, 1}_{m-f-l}, \underbrace{0, \hdots, 0}_l)$. Note that  $R(1,1)$ is the rule $(2,1,\ldots,1,0)$ encountered earlier.

The result in \cite[Theorem 5]{DBLP:conf/atal/Kenig19} along with \cite[Lemma 6]{dey2016exact} provide the following (incomplete) classification of the \PWPV~problem.
\begin{theorem}\label{kenig-thm} 
{\rm \cite{DBLP:conf/atal/Kenig19,dey2016exact}}
Let $r$ be a pure positional scoring rule. Then the following statements hold.
\begin{itemize}
    \item If $r$ is $2$-valued or if $r$ is the rule $R(1,1)$, then  the \emph{\PWPV} \ problem is in \emph{\PTIME}.
    
    \item If $r$ is a differentiating rule, then the \emph{\PWPV} \ problem is in \emph{\NP}-complete.
    
    \item If $r$ is a non-differentiating $p$-valued rule, where $p \geq 3$, other than $R(f,l)$ with $f+l>2$, then the \emph{\PWPV} \ problem w.r.t.\ $r$ is \emph{\NP}-complete.
    
    \item If $r$ is an non-differentiating rule, such that all scoring vectors have at least four distinct values, then the \emph{\PWPV} \ problem w.r.t.\ $r$ is \emph{\NP}-complete.
\end{itemize}
\end{theorem}
The complexity of the \PWPV~problem remains open for the rules  $R(f,l)$ with $f+l>2$.
Since \PWDV~is a special case of \PWPV, the \PTIME~results in Theorem \ref{kenig-thm} also hold for the \PWDV~problem. Note that the \PTIME~result for $2$-valued rules generalises an earlier result in \cite{baumeister2012campaigns}, which established that the \PWDV~problem w.r.t.\ $t$-approval is in \PTIME. 
Moreover, the partial profile constructed in the NP-hardness proof for all differentiating rules in \cite[Lemma 6]{dey2016exact} and  for $p$-valued rules with $p \geq 4$ in \cite[Lemma 13]{DBLP:conf/atal/Kenig19} has only doubly-truncated ballots. Therefore, the \PWDV~problem w.r.t.\ $p$-valued rules with  $p \geq 4$ is also \NP-complete. We further note that the NP-hardness for non-differentiating unbounded rules with scoring vectors containing at least four distinct values is obtained as a corollary to \cite[Lemma 14]{DBLP:conf/atal/Kenig19}. The proof of this lemma implicitly uses doubly-truncated profiles. This is a generalisation of an earlier result  in \cite{betzler2011unweighted,davies2011complexity}, which established that the \PWDV~problem is \NP-complete for Borda count. Thus,  the existing results for the \PWDV~problem can be summarised as follows.
\begin{theorem} \label{thm:PWDV_old} 
\emph{\cite{betzler2011unweighted,davies2011complexity,baumeister2012campaigns,dey2016exact,DBLP:conf/atal/Kenig19}}
The following are true.
    \begin{itemize}
        \item If $r$ is a $2$-valued rule or $r$ is  the rule $R(1,1)$, then the \emph{\PWDV}~problem w.r.t.\ $r$ is in \emph{\PTIME}.
        \item If $r$ is a $3$-valued rule other than $R(f,l)$ with $f+1 > 2,$ or $r$ is a $p$-valued rule, where $p \geq 4$, or $r$ is an unbounded rule with scoring vectors containing at least four distinct score values,  then the \emph{\PWDV}~problem w.r.t.\ $r$ is \emph{\NP}-complete.
    \end{itemize}
\end{theorem}
Therefore, the complexity of the \PWDV~problem w.r.t.\ unbounded rules with scoring vector containing three distinct score values remains open. 
\\
Since \PWTV~and \PWBV~are special cases of \PWDV, the \PTIME~results in Theorem \ref{kenig-thm} and in Theorem \ref{thm:PWDV_old} also hold for \PWTV~and \PWBV~problems. The complexity of these two problems w.r.t $p$-valued rules, where $p \geq 3$, other than the rule $R(1,1)$, and unbounded rules remain open. We settle the complexity of \PWTV~and \PWDV~w.r.t.\ a broad group of unbounded rules (such that all scoring vectors containing at least three distinct values and satisfies some additional properties)  and a restricted group of $3$ valued rules. This implies the NP-completeness of \PWDV~problem w.r.t.\ the same group of unbounded rules.
\subsection{Maximum partial score}

Before presenting the reductions, we introduce a few general notions which will help us reason about partial preferences beyond partial chains, namely doubly-truncated and partitioned preferences.

\begin{definition}
Let $C$ be a set of candidates and $\boldsymbol{P}$ be a partial profile.
\begin{itemize}
    \item We say that a candidate $c' \in C$ is \emph{fixed} in a partial vote $p \in \boldsymbol{P}$ if $c$ has the same position in all extensions $p^*$ of $p_i$.
    
    \item We say that a candidate $c' \in C$ is \emph{fixed} in the partial profile $\boldsymbol{P}$ if $c'$ is fixed in every vote in $\boldsymbol{P}$, i.e., for every $p_i \in \boldsymbol{P},$ there is an integer $b_i$ such that $c$ has position $b_i$ in every extension $p^*_i$ of $p_i.$ (Note that, in general, $b_i$ depends on $p_i$.)
    
    \item We say that a position $b$ in a partial vote is \emph{available} if there is no fixed candidate in $b$, i.e., there exists no candidate $c' \in C$ such that $c'$ is in position $b$ in all extensions $p^*$ of $p$.
    
    \item Let fixed$_{\boldsymbol{P}}(c')$ be the total score made by $c' \in C$ from those votes in $\boldsymbol{P}$ where  $c'$ is fixed.
\end{itemize}
\end{definition}
Let $\boldsymbol{P}$ be a doubly-truncated profile. A candidate is \emph{fixed} in $\boldsymbol{P}$ if for every  $p_i \in P,$ candidate $c$ is in the top or in the bottom.

We need the notion of \emph{maximum partial score}, introduced in \cite{DBLP:journals/jcss/BetzlerD10}, to reason about the completions of doubly-truncated profiles. Since we consider only doubly-truncated profiles in this section, our discussion will be focused on these kinds of votes. However, we note that the notion of maximum partial score can be used for any partial profile where the distinguished candidate is fixed. Before presenting the notion formally,we present an example. Consider a doubly-truncated profile which contains the following partial votes.

\begin{align*}
p_1:& e_3 \succ c \succ e_1 \succ \{ e_2, e_4, e_7\} \succ e_6 \\
p_2:& e_1 \succ \{ e_2, e_4, e_7, e_3\} \succ e_6 \succ c.
\end{align*}
The candidate $c$ is fixed in both $p_1$ and $p_2$. Therefore, $c$ is fixed in the profile. Candidate $e_3$ is fixed in $p_1$ but not in $p_2$.
\\
Let $\boldsymbol{P}$ be a partial profile where candidate $c$ is fixed. Let $\boldsymbol{Q}$ be a total profile. 
Let $\overline{\boldsymbol{P}}$ be an extension of $\boldsymbol{P}$. The score of $c$ in $\overline{\boldsymbol{P}} \cup \boldsymbol{Q}$, namely, $s(\overline{\boldsymbol{P}} \cup \boldsymbol{Q},c) = \lambda$. 
For a candidate $c' \neq c,$ the \emph{maximum partial score} of $c'$, denoted $\maxpartial{c'}$, is defined as $\maxpartial{c'} = \lambda - s(c', \boldsymbol{Q}).$ Intuitively, there exists no completion  $\overline{\boldsymbol{P}}$ of $\boldsymbol{P}$ such that $c$ is a possible winner in $\overline{\boldsymbol{P}} \cup \boldsymbol{Q}$ and $c'$ has a score more than $\maxpartial{c'}$ in $\overline{\boldsymbol{P}}$.

Since, in a doubly-truncated vote, all the candidates which are not in the top or in the bottom are not ordered, the concept of maximum partial score makes it convenient to reason about the completions of such votes. Furthermore, it helps us define the \emph{tightness} property \cite{DBLP:journals/jcss/BetzlerD10}, which will be extremely helpful in proving the correctness of our reductions.

\begin{definition}
Let $C$ be a set of candidates, $c \in C$ be a distinguished candidate,  $\boldsymbol{P}$ be a partial profile such that $c$ is fixed in $\boldsymbol{P}$, and $\boldsymbol{Q}$ be a total profile.  We say that $\boldsymbol{P} \cup \boldsymbol{Q}$ has the \emph{tightness} property if the sum of the score values of all the available positions in all the partial votes in $\boldsymbol{P}$ is equal to the quantity $\sum\limits_{c' \in C \setminus \{c\}} \maxpartial{c'} - \text{fixed}_{\boldsymbol{P}}(c')$.
\end{definition}

The following proposition is quite obvious.

\begin{proposition}
\label{prop:tightness}
Let $C$ be a set of candidates, $c \in C$ be a distinguished candidate, $\boldsymbol{P}$ be a partial profile in which $c$ is fixed, and $\boldsymbol{Q}$ be a total profile. Let $\overline{ \boldsymbol{P}}$ be an extension of $\boldsymbol{P}$  such that $c$ is a winner in $\overline{\boldsymbol{P}} \cup \boldsymbol{Q}$. If $\boldsymbol{P} \cup \boldsymbol{Q}$ has the tightness property, then $s(\overline{\boldsymbol{P}} \cup \boldsymbol{Q} , c') = \maxpartial{c'}$
for every candidate $c' \in C \setminus \{c\}$.
\end{proposition}
\begin{proof}
Otherwise, there exists another candidate $c'' \in C \setminus \{c\}$ which makes more than $\maxpartial{c''}$ and, thus, defeats $c$.
\end{proof}

The reductions in this section require construction of partial profiles where candidates have certain pre-specified maximum partial scores. There is a polynomial time algorithm \cite[Lemma 1]{DBLP:journals/jcss/BetzlerD10} to construct votes to realise the maximum partial scores of the candidates. We restate this fact as follows.

\begin{lemma} {\rm{\cite{DBLP:journals/jcss/BetzlerD10}} }
\label{lem:BD}
Given a scoring rule $r$ with the scoring vector $(s_1, \hdots, s_m)$, a set $C$ of $m$ candidates with distinguished candidate $c \in C,$ a value $\mu(c')$, for all $c' \in C \setminus \{c\},$ and a partial profile $\boldsymbol{P}$ where the following properties hold.
\begin{enumerate}
    \item Candidate $c$ is fixed in $\boldsymbol{P}.$
    \item For every $c' \in C \setminus \{c\},$ the value $\mu (c')$ can be written as a sum of at most $| \boldsymbol{P}|$ integers from $s_1, \hdots, s_m.$ Formally, $\mu(c') = \sum\limits_{j=1}^{m} n_j s_j$ where $n_j \in \mathbb{N}_0$ denotes how often the score value $s_j$ is added. Moreover, $\sum\limits_{j=1}^{m} n_j \leq |\boldsymbol{P}|,$
    
   \item There is a dummy candidate $w$, such that $w$ cannot beat the distinguished $c$ in any extension.
\end{enumerate}
Then, a set $\boldsymbol{Q}$ of total votes can be constructed in time polynomial in $|\boldsymbol{P}|$ and $m$, such that for all $c' \in C \setminus \{c\},$ we have $\maxpartial{c'} = \mu (c')$.
\end{lemma}

We conclude with a proposition on a linear combination of two distinct numbers which will help us prove correctness of the reductions.
\begin{proposition}
\label{prop:lincomb_co_prime}
Let $a_1$ and $a_2$ be two distinct numbers. Let $n_1$ and $n_2$ be two natural numbers and $S = n_1 a_1 + n_2 a_2.$ There exists no  $n_1 \neq n_3 $ and $n_2 \neq n_4$, such that $n_1 + n_2 = n_3 + n_4$ and $n_3 a_1 + n_4 a_2 = S.$
\end{proposition}
\begin{proof}
Suppose there exists $n_3 \neq n_1$ and $n_4 \neq n_2$ such that $n_1 + n_2 = n_3 + n_4$ and $n_3 a_1 + n_4 a_2 = S.$\\
By hypothesis, $n_1 = n_3 + n_4 - n_2.$ Therefore,
\begin{align}
    & n_1 a_1 + n_2 a_2 =  n_3 a_1 + n_4 a_2 \nonumber\\
   \implies~ & (n_3 + n_4 - n_2) a_1 + n_2 a_2 = n_3 a_1 + n_4 a_2 \nonumber\\
   \implies~ & n_2 a_2 - n_2 a_1 = n_4 a_2 - n_4 a_1 \nonumber\\ 
   \implies~ & n_2 ( a_2 - a_1) = n_4 (a_2 - a_1) \nonumber
\end{align}
Since $a_1 \neq a_2$, it must be the case that $n_2 = n_4$ which is a contradiction.
\end{proof}
\subsection{Hardness results for \PWTV}
Recall, that for a scoring vector of length ${m}$, having $m'$ distinct score values, the function $\ell (m, j)$ returns the number of times the distinct score value $a_j$, for $1 \leq j \leq m'$, repeats in a block. Schematically, such a scoring vector can be represented as 

$$\left( \underbrace{a_1, \hdots, a_1}_{\ell (m,1)}, \underbrace{a_{2}, \hdots, a_{2}}_{\ell(m , 2)}, \hdots, \underbrace{ a_{m'}, \hdots, a_{m'}}_{\ell (m, m')} \right)$$
In the reduction below, we consider an unbounded rule $r$ with the following properties. There exists a polynomial $g(u)$ with the property that for all $u$, every scoring vector $\boldsymbol{s}_m$ of $r$ with length $m = g(u)$ has $m' \geq 3$ distinct score values, and if the three smallest score values are $a_{m'-2} > a_{m'-1} > a_{m'}$ , it holds that $m - \ell (m, m'-1) - \ell (m,m') \geq 3u$. Borda count is an obvious example. The lexicographic scoring rule given by $(2^{m}, 2^{m-1}, \hdots, 1)$ for $m$ candidates is another example of the scoring rules considered in this section.
\begin{reduction}
\label{red:PWTTB_ubval}
Let $(\mathcal{X}, \mathcal{Y}, \mathcal{Z}, \mathscr{S})$ be a 3DM instance where and $\mathscr{S} = \{S_1,  \hdots S_t\} \subseteq \mathcal{X} \times \mathcal{Y} \times \mathcal{Z}$ such that $S_i = ( x_{i_1}, y_{i_2}, z_{i_3} )$, for $1 \leq i \leq t$ and $| \mathcal{I}|.$  Let $\boldsymbol{s}_{m}$ be the scoring vector of length $m = g(q)$. The scoring vector has at least three distinct values, i.e., $m' \geq 3$ there are values $a_{m'-2}, a_{m' - 1}, \text{ and } a_{m'}$ such that each of the lengths $\ell (m, m'-1)$ , and $\ell (m, m')$ are fixed and at least one.  We construct an instance of the \PWTV~problem as follows.
\begin{enumerate}
    \item The set of candidates is $C = X \cup Y \cup Z \cup  \{c, w\} \cup H$ where $X$, $Y$, and $Z$ contains candidates corresponding to the elements in $\mathcal{X}, \mathcal{Y}, \text{ and } \mathcal{Z}$ respectively. These candidates are called \emph{element candidates}. The set $H$ contains dummy candidates such that $|H| = m - 3q - 2.$
    
    \item We construct the partial profile $\boldsymbol{P}$ as follows.
    \begin{itemize}
        \item Let the set $H$ be partitioned into $H', H_{m'-1}, \text{ and } H_{m'}$ such that $|H_j| = \ell (m, j) -1$, for $j = \{m'-1, m' \}$.  $H' = H \setminus (H_{m'-1} \cup H_{m'}).$
        For each $S_i = (x_{i_1}, y_{i_2}, z_{i_3} )$, let $C_i' = C \setminus \left( \{x_{i_1}, y_{i_2}, z_{i_3} \} \cup H_{m'-1} \cup H_{m'} \right)$ and $\overrightarrow{C_i'}$ be such that $c \succ w$, i.e., candidate $c$ is always ranked lower than $w$.
        \begin{align*}
            p'_i &= \overrightarrow{C_i'} \succ  x_{i_1} \succ y_{i_2} \succ \overrightarrow{H}_{m'-1} \succ z_{i_3} \succ \overrightarrow{H}_{m'} 
            \\
            p_i &= \overrightarrow{C_i'} \succ ( \{ x_{i_1}, y_{i_2}, z_{i_3} \}  \cup H_{m'-1} \cup {H}_{m'} )
        \end{align*}
        \item $\boldsymbol{P} = \bigcup\limits_{i=1}^{l} p_i$ is profile where each vote is top-truncated.
        \\
        In the vote $p_i,$ note that all the candidates except $\{x_{i_1}, y_{i_2}, z_{i_3} \} \cup H_{m'-1} \cup H_{m'}$ are fixed. In other words, positions $| C'_i| + 1$ though $m$ are available.
        
        $\boldsymbol{P}' = \bigcup\limits_{i=1}^{l} p_i'$ is a total profile. Moreover, each $p'_i$ extends $p_i.$ Let $s(\boldsymbol{P}', c) = \lambda_{\boldsymbol{P}'}.$ Observe that $s(\boldsymbol{P}', w) < \lambda_{\boldsymbol{P}'}$ since $w$ is in a position greater $c$ in all $\overrightarrow{C_i'}$, for $ 1 \leq i \leq t$.
    \end{itemize}

      \item For an element candidate $e  \in X \cup Y \cup Z,$ let  $f_{e}$ denote the number of triples in $\mathscr{S}$ containing the element of the {\sc 3DM} instance corresponding to candidate $e$. By construction of the partial profile, candidate $e$ is not fixed in $f_{e}$ votes. Let fixed$_{\boldsymbol{P}}(c')$ be the total score made by $c' \in C$ from those votes in $\boldsymbol{P}$ where  $c'$ is fixed.
 Therefore, for any candidate $e \in X \cup Y \cup Z,$ we have $\text{fixed}_{\boldsymbol{P}} (e) = \sum\limits_{i=1}^{t- f_{e}} s_{k_i}$, where $1 \leq k_i \leq m$ is the position of $e$ in a vote where it is fixed.
      
     \item  Consider the following.

    \begin{itemize}
        \item For all $x \in X $, we have 
        $\mu(x) = a_{m'} + (f_x - 1) a_{m'-2} + \text{fixed}_{\boldsymbol{P}}(x).$ 
            
        \item For all $y \in Y $, we have 
        $\mu(y) = a_{m'-2} + (f_y - 1) a_{m'-1} + \text{fixed}_{\boldsymbol{P}}(y).$
    
        \item For all $z \in Z $, we have 
       $\mu(z) = a_{m'-1} + (f_z - 1) a_{m'} + \text{fixed}_{\boldsymbol{P}}(z).$
        
         \item For all $h \in H_j $, we have 
        $\mu(h) = t(a_{j})$  where $j= \{m'-1, m'\}$.
        
         \item For all $h' \in H' $, we have 
        $\mu(h) = \text{fixed}_{\boldsymbol{P}} (h')$.
         
         \item $\mu(w) \geq t a_1.$
        
    \end{itemize}
    
    \item We verify that the profile $\boldsymbol{P}$, and, for all $c' \in C \setminus \{c\}$, the number 
    $\mu(c')$, as specified above, satisfy the properties required by Lemma \ref{lem:BD}.
    \begin{itemize}
        \item Property 1: By the construction of the votes in the reduction, this property is satisfied.
        \item Property 2: For all $e \in X \cup Y \cup Z$, the number $\mu(e)$ is the sum of 
        $(t -f_{e}) + (f_{e} -1) + 1 = t$ score values.
     For all $h \in H$, property 2 is satisfied trivially.
        \item Property 3: Candidate $w$ is fixed in $\boldsymbol{P}$, and in every vote, has a position greater than that of $c$, and therefore, can never defeat $c$ in any extension.
    \end{itemize}
    Therefore, by the lemma, there is a total profile $\boldsymbol{Q}$, which can be constructed in time polynomial in $| \boldsymbol{P} |$ and $m$, such that $\maxpartial{c'} = \mu (c')$, for all $c' \in C \setminus \{c\}$.

    \item We let $C,$ the profile $\boldsymbol{V} =  \boldsymbol{P} \cup \boldsymbol{Q}, \text{ and } c$ be the input to the \PWTV~problem.
\end{enumerate} 
\end{reduction}

\begin{proposition}
\label{prop:tight_PWTTB_ubval}
The profile $\boldsymbol{P} \cup \boldsymbol{Q}$ in Reduction \ref{red:PWTTB_ubval} has the tightness property.
\end{proposition}
\begin{proof}
 Recall that $|X| = |Y| = |Z| = q.$ Therefore,

$$\sum\limits_{x \in X}f_x = t \implies \sum\limits_{x \in X}(f_x - 1) = \sum\limits_{x \in X}f_x - \sum\limits_{x \in X} 1 = t - q.$$ 

Similarly, $\sum\limits_{y \in Y}f_y = \sum\limits_{z \in Z}f_z = t$ and $\sum\limits_{y \in Y}(f_y - 1) = \sum\limits_{z \in Z}(f_z - 1) = t-q.$
\\
We focus only on the available positions in the votes of $\boldsymbol{P}$ and the scores the candidates can make in these positions. Thus, for a candidate $e \in X, Y, \text{ and } Z$ we consider the score $\maxpartial{e} - \text{fixed}_{\boldsymbol{P}} (e).$ All candidates $h' \in H'$ and the candidate $w$ are fixed in $\boldsymbol{P}$ and therefore, contributes nothing to the sum. Whereas, all candidates $h \in H_{m'-1} \cup H_{m'}$ are not fixed in any vote in $\boldsymbol{P}$, i.e. $\text{fixed}_{\boldsymbol{P}}(h) = 0.$
 Therefore, the sum of the maximum scores which the candidates can make in the available positions is\\

\begin{align}
& \hspace{-0.75cm} \sum\limits_{c' \in (X \cup Y \cup Z \cup H_{m'-1} \cup H_{m'})} \left( \maxpartial{c'}  - \text{fixed}_{\boldsymbol{P}}(c') \right)
\nonumber \\
   & = \sum\limits_{x \in X} \left( \maxpartial{x} - \text{fixed}_{\boldsymbol{P}}(x) \right)
    + \sum\limits_{y \in Y} \left( \maxpartial{y} - \text{fixed}_{\boldsymbol{P}}(y) \right) \nonumber\\
   & \hspace{1cm} +  \sum\limits_{z \in Z} \left( \maxpartial{z} - \text{fixed}_{\boldsymbol{P}}(z)  \right) + \sum\limits_{h \in H_{m'-1}}\maxpartial{h} 
   + \sum\limits_{h \in H_{m}} \maxpartial{h} 
   \nonumber \\
   & = \sum\limits_{x \in X} \left( a_{m'} + (f_x - 1) a_{m'-2} \right)
    + \sum\limits_{y \in Y} \left( a_{m'-2} + (f_y - 1) a_{m'-1}) \right)
    +  \sum\limits_{z \in Z} \left( a_{m'-1} + (f_z - 1)a_{m'}) \right)\nonumber\\
   & \hspace{1cm}  + \sum\limits_{h \in H_{m'-1}}t(a_{m'-1}) 
   + \sum\limits_{h \in H_{m}} t(a_{m'}) 
   \nonumber \\
   & = q(a_{m'}) + (t-q)(a_{m'-2}) + q(a_{m'-2}) + (t-q)(a_{m'-1}) \nonumber \\ 
   & \hspace{1cm} + q(a_{m'-1}) + (t-q)(a_{m'}) + (\ell(m, m'-1) - 1)t(a_{m'-1}) + (\ell(m, m') - 1)t(a_{m'}) 
   \nonumber \\
   & = \ell(m, m') t(a_{m'}) +  \ell(m, m'-1) t (a_{m'-1}) +  t(a_{m'-2}) \nonumber \\
   & = t \left( a_{m'-2} + \ell(m, m'-1)  a_{m'-1} +  \ell(m, m') a_{m'} \right).
   \label{eq:sum_partial_PWDTB_ub}
\end{align}

Recall, that there are $t$ votes in $\boldsymbol{P}$, one corresponding to every triple in $\mathscr{S}.$ 
Therefore, the sum of the score values of the available positions in the $t$ votes , namely
\begin{itemize}
    \item position $ m - \ell(m, m'-1) - \ell(m, m') $, with score value $a_{m'-2}$,
    \item  positions $ m - \ell(m, m'-1) - \ell(m, m') + 1 $ through 
    $  m -  \ell(m, m') $, each with score value  $a_{m'-1}$, and
    \item positions $m - \ell(m, m') + 1 $ through $m$, with score value $a_{m'}$
\end{itemize} 
 is $t \left( a_{m'-2} + \ell(m, m'-1) (a_{m'-1}) + \ell(m', m) a_{m'} \right)$
which is the same as in (\ref{eq:sum_partial_PWDTB_ub}).
\end{proof}
\begin{proposition}
\label{prop:x_y_in_a_m-2_PWDTB_ub}
In Reduction \ref{red:PWTTB_ubval}, for all completions $\overline{\boldsymbol{P}}$ of $\boldsymbol{P}$, if $c$ is a possible winner in $\overline{\boldsymbol{P}} \cup \boldsymbol{Q}$
then the candidate which is in position 
$ m - \ell(m , m'-1) - \ell(m, m') $, with score value $a_{m'-2}$ is an element candidate of $X \cup Y.$
\end{proposition}
\begin{proof}
By the construction of the partial profile in Reduction \ref{red:PWTTB_ubval}, there are $t$ votes in $\boldsymbol{P}$. By Proposition \ref{prop:tight_PWTTB_ubval}, the profile $\boldsymbol{P} \cup \boldsymbol{Q}$ has the tightness property. Therefore, the sum of the maximum scores which the elements in $X \cup Y$ can make from the available positions in all the votes in $\boldsymbol{P}$ is
\begin{align}
   & \hspace{-2cm} \sum\limits_{x \in X} \left( \maxpartial{x} - \text{fixed}_{\boldsymbol{P}}(x) \right)
   +  \sum\limits_{y \in Y} \left( \maxpartial{y} - \text{fixed}_{\boldsymbol{P}}(y) \right) \nonumber\\
    & =  q(a_{m'}) + (t-q)(a_{m'-2}) + q(a_{m'-2}) + (t-q)(a_{m'-1}) \nonumber \\
    & = q(a_{m'}) + t(a_{m'-2}) + (t-q)(a_{m'-1}).
    \label{eq:sum_partial_x_y_PWDTB_ub}
\end{align}
Let $\overline{\boldsymbol{P}}$ be a completion of $\boldsymbol{P}$. In each vote $\overline{p} \in \overline{\boldsymbol{P}},$ an element of $X \cup Y$ can gain $a_{m'-2}$, $a_{m'-1}$, or $a_{m'}$ points. 
If there exists a vote in $\overline{\boldsymbol{P}}$ such that an element candidate of $X \cup Y$ is not in position $ m - \ell(m , m'-1) - \ell(m, m') $ the total score will be strictly less than that in \ref{eq:sum_partial_x_y_PWDTB_ub}, and therefore in violation of the tightness property.
\end{proof}
\begin{lemma}
\label{lem:PWTTB_ub_val}
Let $r$ be an unbounded rule such that there exists a polynomial $g(u)$ with the property that for all $u$, every scoring vector $\boldsymbol{s}_m$ of $r$, with length $m = g(u)$, has $m' \geq 3$ distinct score values, and if the three smallest score values are $a_{m'-2} > a_{m'-1} > a_{m'}$ , it holds that $m - \ell (m, m'-1) - \ell (m,m') \geq 3u$. Then reduction \ref{red:PWTTB_ubval} is a polynomial-time reduction of {\sc 3DM } to \emph{\PWTV}~w.r.t.\ $r$.
\end{lemma}
\begin{proof}
    Given a 3DM instance $\mathcal{I} = (\mathcal{X}, \mathcal{Y}, \mathcal{Z}, \mathscr{S}),$ we construct a \PWTV~ instance, $(C, \boldsymbol{V} = \boldsymbol{P} \cup \boldsymbol{Q}, c),$ according to the above reduction. We let $u =|\mathcal{I}| = q.$
    \\
    First, we prove the "$\impliedby$'' direction.
    Assume that the \PWPC~instance $(C, \boldsymbol{V} = \boldsymbol{P} \cup \boldsymbol{Q}, c)$ obtained above is a positive one. Therefore, there exists a total profile $\boldsymbol{P}^* = \bigcup_{i=1}^{l} p^*_i$ such that  
    \begin{itemize}
        \item for all $ 1 \leq i \leq t, p^*_i \text{ extends } p_i$
        \item $c$ is a possible winner and its score remains the same in all extensions (since by construction, $c$ is a fixed candidate in $\boldsymbol{P}$).
    \end{itemize}
 By Proposition \ref{prop:tight_PWTTB_ubval}, $\boldsymbol{P}$ has the tightness property. In the following, use tightness and  other properties of $\boldsymbol{P}$ to reason about $\boldsymbol{P}^*$. Recall the schematic representation of the scoring vector is
 $$\left(  \underbrace{ \hdots, a_{m'-2}}_{\geq 3q}, \underbrace{ a_{m'-1}, \hdots, a_{m'-1}}_{\ell (m, m'-1)}, \underbrace{ a_{m'}, \hdots, a_{m'}}_{\ell (m, m')} \right)$$
    \begin{enumerate}
        
        \item By construction and the maximum partial scores set in the reduction, all candidates $h \in H_{m'}$ must be in a position with score value $a_{m'}.$ If any candidate $h \in H_{m'}$ is in a position with score value greater than $a_{m'}$, it will defeat $c$. Without loss of generality, assume that they are in positions 
        positions $m - \ell(m, m') + 2 $ through $m$.
        
         \item By Proposition \ref{prop:x_y_in_a_m-2_PWDTB_ub}, in all the votes, the candidate in the position $ m - \ell(m, m') - \ell(m, m' -1) $, which has a score value of $a_{m'-2}$, is an element candidate of $X \cup Y.$ \label{point:posa_PWDTB_ub}

        \item Observe that no candidate $h$ in $H_{m'-1}$ can be in position $ m - \ell(m, m') + 1 $. \\
        For, if there exists such an $h$, then, by tightness, it would have to be in position $ m - \ell(m, m') $ at least once. By \ref{point:posa_PWDTB_ub}, this is not possible.
        
        \item Observe that, by Proposition \ref{prop:lincomb_co_prime},  for $1 \leq i \leq t,$ in any completion of a vote $p_i$, for $S_i = (x_{i_1}, y_{i_2}, z_{i_3} )$, the element candidate corresponding to $y_{i_2}$ has to be in a position with score value $a_{m'-2}$ or $a_{m'-1}.$ Otherwise, since the total number of votes is fixed, it would violate tightness. 
        
        \begin{enumerate}
            \item Moreover, by the maximum partial scores set in the construction, for $1 \leq i \leq q,$ every $y_i$ scores $a_{m'-2}$ exactly once. More than once, $y_{i}$ defeats $c$ and less than once violates tightness. Let these votes be $p^*_{k_1}, \hdots, p^*_{k_q}$, where $1 \leq k_i \leq t$ and $K = \{k_i | 1 \leq i \leq q \}.$
            
            \item For $1 \leq i  \leq t \text{ and } i  \notin K,$ in $p^*_{i}$, candidate $y_{i_2}$ is in position $m - \ell(m, m') -  \ell(m, m'-1) + 1 .$
            \label{point:pos_a-2_+1_PWDTB_ub}
        \end{enumerate}
        
        \item Thus, for $1 \leq i \leq t \text{ and } i \notin K,$ in the vote $p^*_{i}$,
        
        \begin{enumerate}
            \item by \ref{point:posa_PWDTB_ub} and \ref{point:pos_a-2_+1_PWDTB_ub}, candidate $x_{i_1}$ is in position $ m - \ell(m, m') - \ell(m, m'-1) $.
            
            \item $z_{i_3}$ must be in $m - \ell(m, m') + 1.$
        \end{enumerate}

        \item  By the tightness property, for $i \in K$, in $p^*_i,$ candidate $z_{i_3}$ must be in position $m - \ell(m , m') - \ell(m, m'-1) +1.$ By the maximum partial scores set for $z \in Z,$ the element candidates of $Z$ in these $q$ votes must be distinct.
        
        \item By the tightness property, each of the $q$ candidates in $X$ have to score $a_{m'}$ points exactly once. Therefore it must be true that the element candidates of $X$ in position $m - \ell (m, m') + 1$ in the votes $p^*_i$, for $i \in K$, are distinct.
        
        \item Therefore, the set $\{S_i | i \in K\}$ must form a cover for $\mathcal{X} \cup \mathcal{Y} \cup \mathcal{Z}.$
    \end{enumerate}
    For the `$\implies$' direction, let $(\mathcal{X}, \mathcal{Y}, \mathcal{Z}, \mathscr{S})$ be a positive instance. We show that in the \PWDV~ instance, $(C, \boldsymbol{V} = \boldsymbol{P} \cup \boldsymbol{Q}, c),$ constructed in the reduction, $c$ is, indeed, a possible winner. By hypothesis, there is a $\mathscr{S}' \subseteq \mathscr{S}$ and $| \mathscr{S}' | = q.$
    \begin{enumerate}
        \item Complete each vote  $p_i \in \boldsymbol{P}$ to $p^*_i$ as follows.
        \begin{itemize}
            \item $p^*_i = \overrightarrow{C_i'}  \succ  y_{i_2} \succ z_{i_3} \succ \overrightarrow{H}_{m'-1} \succ x_{i_1} \succ \overrightarrow{H}_{m'}$  if $S_i \in \mathscr{S}'$
            
            \item $p^*_i = \overrightarrow{C_i'} \succ x_{i_1}  \succ  y_{i_2} \succ  \overrightarrow{H}_{m'-1} \succ z_{i_3} \succ \overrightarrow{H}_{m'}$ if $S_i \notin \mathscr{S}'$
        \end{itemize}
        
        Let $\boldsymbol{P}^* = \bigcup_{i=1}^{l} p^*_i.$

    \item We verify that the candidates have, indeed, scored no more than the respective maximum partial scores in the completion $\boldsymbol{P}^*$ of $\boldsymbol{P}$. Recall that the position of $c$ is fixed in both $\boldsymbol{Q}$ and $\boldsymbol{P}$.
    
    \begin{itemize}
            \item For all $x \in X,$ we have 
            $ s(\boldsymbol{P}^*, x)
            = a_{m'} + (f_x - 1) a_{m'-2} + \text{fixed}_{\boldsymbol{P}}(x).$
            
             \item For all $y \in Y,$ we have $s(\boldsymbol{P}^*, y)
            = a_{m'-2} + (f_y - 1) a_{m'-1} + \text{fixed}_{\boldsymbol{P}}(y).$
            
            \item For all $z \in Z,$ we have $s(\boldsymbol{P}^*, z)
            = a_{m'-1} + (f_z - 1) a_{m'} + \text{fixed}_{\boldsymbol{P}}(z).$
            
            \item For all $h \in H_{m'-1},$ we have $
            s(\boldsymbol{P}^*, h)
            = t(a_{m'-1}).$
            
            \item For all $h \in H_{m'},$ we have $
            s(\boldsymbol{P}^*, h)
            = t(a_{m'}).$

             \item The positions of all the candidates in $H'$, as constructed in the reduction, are fixed and therefore their scores do not change, i.e., for all $h' \in H',$ we have
            $ s(\boldsymbol{P}^*, h')
            = \text{fixed}_{\boldsymbol{P}}(h')$
            
            \item By construction, we have $s(\boldsymbol{P}^*, w) < s(\boldsymbol{P}^*, c).$
        \end{itemize}
    Therefore, $c$ is a possible winner.
    \end{enumerate}
    \vspace{-2pc}
\end{proof}

The following lemma is a direct consequence of the Lemma \ref{lem:PWTTB_ub_val}.

\begin{lemma}
\label{lem:PWTTB_pval}
Let $r$ be a $p$-valued positional scoring rule, where $p \geq 3$, there exists a polynomial $g(u)$ with the property that for all $u$, every scoring vector $\boldsymbol{s}_m$ of $r$, with length $m = g(u)$, has $m' = p$ distinct score values, and if the three smallest score values are $a_{m'-2} > a_{m'-1} > a_{m'}$ , it holds that $m - \ell (m, m'-1) - \ell (m,m') \geq 3u$. Then \emph{\PWTV}~problem w.r.t.\ $r$ is NP-complete. 
\end{lemma}
\begin{proof}
Observe that $p$-valued positional scoring rule, where $p \geq 3$, is sufficient for the construction in Reduction \ref{red:PWTTB_ubval}. In particular, to specify the value $\mu(c')$ for each candidate $c' \in C \setminus \{c\}$ we need at least three distinct score values which we always have.
\end{proof}
\subsection{Hardness results for \PWBV}
In this section, we prove \NP-completeness for the \PWBV~problem w.r.t a restricted group of unbounded rules. We consider an unbounded rule $r$ with the following properties. There exists a polynomial $g(u)$ with the property that for all $u$, every scoring vector $\boldsymbol{s}_m$ of $r$ with length $m = g(u)$ has at least three distinct score values, and if the largest three score values are $a_1 > a_2 > a_3$ , it holds that $m - \ell (m, 1) - \ell (m,2) \geq 3u$. Besides the Borda count, examples include broad families of scoring rules with scoring vectors in which the fist $k$ score values are distinct. Schematically, such a scoring vector can be represented as follows
$$\left(  \underbrace{a_{1}, \hdots, a_{1}}_{\ell(m, , 1)}, \underbrace{ a_{2}, \hdots, a_{2}}_{\ell (m, 2)}, \underbrace{ a_{3},  \hdots}_{\geq 3u} \right)$$

\begin{reduction}
\label{red:PWBTB_ubval}
The reduction is similar to Reduction \ref{red:PWTTB_ubval}. We only provide the key steps below.
\begin{itemize}
    \item The set of candidates is $C = X \cup Y \cup Z \cup  \{c, w\} \cup H$ where $X$, $Y$, and $Z$ contains candidates corresponding to the elements in $\mathcal{X}, \mathcal{Y}, \text{ and } \mathcal{Z}$ respectively. These candidates are called \emph{element candidates}. The set $H$ contains dummy candidates such that $|H| = m - 3q - 2.$

    \item We construct the partial profile $\boldsymbol{P}$ as follows.
    \begin{itemize}
        \item Let the set $H$ be partitioned into $H_1, H_2, \text{ and } H'$ such that $|H_j| = \ell (m, j) -1$, for $j \in \{ 1,2 \}$. $H' = H \setminus (H_1 \cup H_2).$
        For each $S_i = (x_{i_1}, y_{i_2}, z_{i_3} )$, let $C_i' = C \setminus \left( \{x_{i_1}, y_{i_2}, z_{i_3} \} \cup H_{1} \cup H_{2} \right)$ and $\overrightarrow{C_i'}$ be such that $c \succ w$, i.e., candidate $c$ is always ranked lower than $w$.
        \begin{align*}
            p'_i &= \overrightarrow{H_1} \succ  x_{i_1} \succ y_{i_2} \succ \overrightarrow{H}_{2} \succ z_{i_3} \succ \overrightarrow{C_i'} 
            \\
            p_i &= \left( H_1 \cup \{ x_{i_1} , y_{i_2} , z_{i_3} \} \cup H_2 \right) \succ \overrightarrow{C_i'}
        \end{align*}
        \item $\boldsymbol{P} = \bigcup\limits_{i=1}^{l} p_i$ is a partial profile where each vote is bottom-truncated.
        \\
        In the vote $p_i,$ note that all the candidates except $\{x_{i_1}, y_{i_2}, z_{i_3} \} \cup H_{1} \cup H_{2}$ are fixed. In other words, positions one though $\ell(m, 1) + \ell (m, 2) + 1$ are available.
        
        $\boldsymbol{P}' = \bigcup\limits_{i=1}^{l} p_i'$ is a total profile. Moreover, each $p'_i$ extends $p_i.$ Let $s(\boldsymbol{P}', c) = \lambda_{\boldsymbol{P}'}.$ Observe that $s(\boldsymbol{P}', w) < \lambda_{\boldsymbol{P}'}$ since $w$ is in a position greater $c$ in all $\overrightarrow{C_i'}$, for $ 1 \leq i \leq t$.
    \end{itemize}
      \item  Now we specify the value $\mu(c')$ for each $c' \in C \setminus \{c\}.$

    \begin{itemize}
        \item For all $x \in X $, we have 
        $\mu(x) = a_{3} + (f_x - 1) a_{1} + \text{fixed}_{\boldsymbol{P}}(x).$ 
            
        \item For all $y \in Y $, we have 
        $\mu(y) = a_{1} + (f_y - 1) a_{2} + \text{fixed}_{\boldsymbol{P}}(y).$
    
        \item For all $z \in Z $, we have 
       $\mu(z) = a_{2} + (f_z - 1) a_{3} + \text{fixed}_{\boldsymbol{P}}(z).$
        
         \item For all $h \in H_j $, we have 
        $\mu(h) = t(a_{j})$  where $1 \leq j \leq m'$.
         
         \item $\mu(w) \geq t a_1.$
        
    \end{itemize}
\end{itemize}
The remaining steps are identical to Reduction \ref{red:PWTTB_ubval}.
\end{reduction}
\begin{proposition}
\label{prop:tight_PWBTB_ubval}
The profile $\boldsymbol{P} \cup \boldsymbol{Q}$ in Reduction \ref{red:PWBTB_ubval} has the tightness property.
\end{proposition}
\begin{proof}
The proof of tightness property for the above construction is similar to the proof of Proposition \ref{prop:tight_PWTTB_ubval}.
\end{proof}
\begin{lemma}
\label{lem:PWBTB_ub_val}
Let $r$ be an unbounded rule such that there exists a polynomial $g(u)$ with the property that for all $u$, every scoring vector $\boldsymbol{s}_m$ of $r$ with length $m = g(u)$ has at least three distinct score values, and if the largest three score values are $a_1 > a_2 > a_3$ , it holds that $m - \ell (m, 1) - \ell (m,2) \geq 3u$. Then reduction \ref{red:PWBTB_ubval} is a polynomial-time reduction of {\sc 3DM } to \emph{\PWTV}~w.r.t.\ $r$.
\end{lemma}
\begin{proof}
The proof is similar to the proof of Lemma \ref{lem:PWTTB_ub_val}.
\end{proof}
The following lemma is a direct consequence of Lemma \ref{lem:PWBTB_ub_val}.
\begin{lemma}
\label{lem:PWBTB_pval}
Let $r$ be a $p$-valued positional scoring rule, where $p \geq 3$, such that there exists a polynomial $g(u)$ with the property that for all $u$, every scoring vector $\boldsymbol{s}_m$ of $r$ with length $m = g(u)$ has $p$ distinct score values, and if the largest three score values are $a_1 > a_2 > a_3$ , it holds that $m - \ell (m, 1) - \ell (m,2) \geq 3u$. Then the \emph{\PWBV}~problem w.r.t.\ $r$ is NP-complete. 
\end{lemma}
\begin{proof}
Observe that $p$-valued positional scoring rule, where $p \geq 3$, is sufficient for the construction in Reduction \ref{red:PWBTB_ubval}. In particular, to specify the value $\mu(c')$ for each candidate $c' \in C \setminus \{c\}$ we need at least three distinct score values which we always have.
\end{proof}
\subsection{Hardness of \PWDV~ w.r.t.\ unbounded rules}
We conclude with the following result for \PWDV~w.r.t.\ to a broad group of unbounded rules. This generalises the existing results in Theorem \ref{thm:PWDV_old}. Putting together the hardness results in Lemma \ref{lem:PWTTB_ub_val} and Lemma \ref{lem:PWBTB_ub_val}, we can state the following theorem.
\begin{theorem} \label{thm:PWDV_new}
    Let $r$ be an unbounded rule that satisfies one of the following conditions:
    \begin{enumerate}
        \item There exists a polynomial $g(u)$ such that for every $u$, the scoring vector $\boldsymbol{s}_m$ of $r$ with $m= g(u)$ has $m' \geq 3$ distinct score values, and if the three smallest score values are $a_{m'-2} > a_{m'-1} > a_{m'}$ , it holds that $m - \ell (m, m'-1) - \ell (m,m') \geq 3u$.
        
        \item There exists a polynomial $g(u)$ such that for every $u$, the scoring vector $\boldsymbol{s}_m$ of $r$ with $m= g(u)$ has at least three distinct score values, and if the largest three score values are $a_1 > a_2 > a_3$ , it holds that $m - \ell (m, 1) - \ell (m,2) \geq 3u$.
    \end{enumerate} 
    Then the \emph{\PWDV}~problem w.r.t.\ $r$ is \emph{\NP}-complete.
\end{theorem}
\begin{proof}
Since \PWTV~and \PWBV~are special cases of \PWDV, the NP-hardness results in Lemma \ref{lem:PWTTB_ub_val} and Lemma \ref{lem:PWBTB_ub_val} also hold for \PWDV.
\end{proof}
Therefore, \PWDV~w.r.t.\ unbounded rules with scoring vectors containing three values such that the length of the block containing the second score value is unbounded remains open.The complexity of \PWDV~w.r.t.\ to all $p$-valued rules except $R(f,l)$, such that $f+1>2$, remains to be established. 


\section{Concluding Remarks}
\begin{table}[ht]
\centering
\begin{tabular}{l|cccccc}  
\toprule
Scoring Rule  & \PW & \PWPC & \PWPV & \PWDV & \PWTV & \PWBV \\
\toprule
Plurality \& Veto& \Ptime  & \Ptime &  \Ptime  &  \Ptime  & \Ptime & \Ptime   \\
\midrule
$2$-valued & NP-c   & \textbf{NP-c}  &  \Ptime   & \Ptime &  \Ptime   & \Ptime  \\
\midrule
$R(1,1)$ & NP-c  & \textbf{NP-c} &  \Ptime  &  \Ptime &  \Ptime   & \Ptime    \\
\midrule
$R(f,l)$, $f+l>2$ & NP-c  & \textbf{NP-c} &  ?  &  ? &  ?  &  ?    \\
\midrule
All other $3$-valued            & NP-c   & \textbf{NP-c}  &  NP-c   & NP-c   & \textbf{NP-c}$^*$[Lem. \ref{lem:PWTTB_pval}]  & \textbf{NP-c}$^*$[Lem. \ref{lem:PWBTB_pval}]     \\
\midrule
$p$-valued,  $p \geq 4$           & NP-c   & \textbf{NP-c}  & NP-c   & NP-c  & \textbf{NP-c}$^*$[Lem. \ref{lem:PWTTB_pval}]  & \textbf{NP-c}$^*$[Lem. \ref{lem:PWBTB_pval}]     \\
\midrule
Unbounded rules & NP-c   & \textbf{NP-c}  &  NP-c    & \textbf{NP-c}$^*$[Thm. \ref{thm:PWDV_new}]  & \textbf{NP-c}$^*$[Lem. \ref{lem:PWTTB_ub_val}]  & \textbf{NP-c}$^*$[Lem. \ref{lem:PWBTB_ub_val}]  \\
\bottomrule 
\end{tabular}
\newline 
\caption{Classification of the PW problem and its various restrictions. The results in boldface have been established in this paper. 
\newline $*$See the respective results for restrictions.
}
         
\label{tab:results}
\end{table}
The contributions in this paper can be summarised as follows.

\begin{itemize}
    \item We obtained a complete classification of the complexity of the \PW~problem on partial chains w.r.t.\ to all pure positional scoring rules. 
Since the classification we obtained is the same as that of the \PW~problem on arbitrary partial orders,  we gave a new,  self-contained (and, in our view, more principled) proof of the original classification theorem for \PW.

    \item We established new \NP-completeness results for the \PW~problem on top-truncated, and bottom-truncated partial orders. These results also hold for the \PW~problem on doubly-truncated partial orders.
    
    \item Our results, together with their comparison to earlier related results in the literature,   are depicted in Table \ref{tab:results}.
\end{itemize}

In terms of future work, it remains an open problem to pinpoint the complexity of the \PW~problem w.r.t.\ rules $R(f,l)$ with $f+l>2$ on doubly-truncated partial orders and on partitioned partial orders. In a different direction, there is a rich body of work on algorithmic problems about manipulation in voting (\PW~is a special case of one of these problems), where computational hardness is regarded as a feature because it provides an obstacle to such manipulation (see \cite{DBLP:reference/choice/ConitzerW16} for a survey). More recent work in this area includes the study of manipulation in voting when only incomplete preferences, expressed as partial orders, are available \cite{DBLP:journals/tcs/DeyMN18}. Furthermore, top-truncated partial orders have been studied in this setting
\cite{DBLP:journals/aamas/MenonL17}. It would be natural to investigate manipulation in voting with partial chains as incomplete preferences.

As a broader agenda, we note that a framework aiming to create bridges between computational social choice and relational databases was introduced in \cite{kimelfeld2018computational} and studied further in \cite{DBLP:conf/pods/KimelfeldKT19}. In that framework, the main concepts are the \emph{necessary answers} and the \emph{possible answers} to queries about winners in elections together with relational context about candidates, voters, and candidates' positions on issues. 
It should be pointed out that the necessary answers to natural database queries may be intractable (co\NP-complete), even w.r.t.\ the plurality rule.
Thus, our work motivates the investigation of the complexity of the necessary answers and the possible answers to queries on partial chains and other restricted classes of partial orders.

\paragraph{Acknowledgement}The work of both authors have been supported by NSF Grant IIS 1814152.

\bibliographystyle{unsrt}  
\bibliography{references}  


\end{document}